\documentclass[12pt]{article}
\usepackage{amsmath, amsthm}
\usepackage{amscd}
\usepackage{amssymb}
\usepackage{array}
\usepackage{mathtools}
\usepackage{color}

\newtheorem{thm}{Theorem}[section]
\newtheorem{theorem}[thm]{Theorem}

\newtheorem{corollary}[thm]{Corollary}

\newtheorem{proposition}[thm]{Proposition}

\theoremstyle{definition}
\newtheorem{definition}[thm]{Definition}
\newtheorem{example}[thm]{Example}

\newtheorem{remark}[thm]{Remark}
\newtheorem{notation}[thm]{Notation}

\newcommand{\algaff}{\mathrm{(alg_\mathrm{aff})}}
\newcommand{\salgaff}{\mathrm{(salg_\mathrm{aff})}}
\newcommand{\schemesaff}{\mathrm{(schemes_\mathrm{aff})}}
\newcommand{\Fl}{\mathrm{Fl}}
\newcommand{\F}{\mathrm{F}}

\newcommand{\be}{\begin{equation}}
\newcommand{\ee}{\end{equation}}
\newcommand{\bea}{\begin{eqnarray}}
\newcommand{\eea}{\end{eqnarray}}
\newcommand{\bean}{\begin{eqnarray*}}
\newcommand{\eean}{\end{eqnarray*}}

\newcommand{\id}{\relax{\rm 1\kern-.28em 1}}
\newcommand{\R}{\mathbb{R}}
\newcommand{\C}{\mathbb{C}}
\newcommand{\Z}{\mathbb{Z}}
\newcommand{\N}{\mathbb{N}}

\newcommand{\bV}{\mathbb{V}}

\newcommand{\bk}{k}

\newcommand{\bP}{\mathbf{P}}

\newcommand{\cA}{\mathcal{A}}

\newcommand{\cL}{\mathcal{L}}
\newcommand{\cO}{\mathcal{O}}
\newcommand{\cG}{\mathcal{G}}

\newcommand{\cF}{\mathcal{F}}

\newcommand{\ep}{\mathcal{E}}

\newcommand{\cM}{\mathcal{M}}
\newcommand{\cV}{\mathcal{V}}

\newcommand{\ri}{\mathrm{i}}

\newcommand{\rGL}{\mathrm{GL}}
\newcommand{\rSU}{\mathrm{SU}}
\newcommand{\rSL}{\mathrm{SL}}
\newcommand{\rSO}{\mathrm{SO}}

\newcommand{\rU}{\mathrm{U}}

\newcommand{\rM}{\mathrm{M}}

\newcommand{\rid}{\mathrm{id}}

\newcommand{\Hom}{\mathrm{Hom}}

\newcommand{\rspan}{\mathrm{span}}

\newcommand{\Gr}{\mathrm{Gr}}
\newcommand{\ber}{\mathrm{Ber}}

\newcommand{\fsl}{\mathfrak{sl}}
\newcommand{\fg}{\mathfrak{g}}
\newcommand{\fh}{\mathfrak{h}}

\newcommand{\fso}{\mathfrak{so}}

\newcommand{\fn}{\mathfrak{n}}
\newcommand{\fp}{\mathfrak{p}}
\newcommand{\fb}{\mathfrak{b}}

\newcommand{\tr}{\mathrm{tr}}

\newcommand{\Ber}{\mathrm{Ber}}
\newcommand{\spec}{\mathrm{Spec}}
\newcommand{\proj}{\mathrm{Proj}}
\newcommand{\sschemes}{\mathrm{(sschemes)}}
\newcommand{\sschemeaff}{\mathrm{ {( {sschemes}_{\mathrm{aff}} )} }}

\newcommand{\alg}{\mathrm{(alg)}}

\newcommand{\sets}{\mathrm{(sets)}}

\newcommand{\salg}{\mathrm{(salg)}}
\newcommand{\varaff}{ \mathrm{ {( {var}_{\mathrm{aff}} )} } }
\newcommand{\svaraff}{\mathrm{ {( {svar}_{\mathrm{aff}} )}  }}

\newcommand{\Ad}{\mathrm{Ad}}

\newcommand{\Lie}{\mathrm{Lie}}

\newcommand{\rhom}{\mathrm{hom}}

\newcommand{\uspec}{\underline{\mathrm{Spec}}}
\newcommand{\uproj}{\mathrm{\underline{Proj}}}

\newcommand{\al}{\alpha}

\newcommand{\de}{\delta}
\newcommand{\dd}{\delta}

\newcommand{\s}{\sigma}
\newcommand{\sig}{\sigma}
\newcommand{\lra}{\longrightarrow}

\newcommand{\Oqgh}{\cO_q\big(G\big/P\big)} 

\begin{document}


\medskip

\centerline{\Large{\bf
The  Segre embedding of the quantum}}

\medskip

\centerline{\Large{\bf conformal superspace} }

\bigskip

\centerline{ R. Fioresi}

\smallskip

\centerline{\it Dipartimento di Matematica, Universit\`{a} di
Bologna }
 \centerline{\it Piazza di Porta S. Donato, 5. 40126 Bologna, Italy.}
\centerline{{\footnotesize e-mail: rita.fioresi@UniBo.it}}
\bigskip

\centerline{ E. Latini}

\smallskip

\centerline{\it Dipartimento di Matematica, Universit\`{a} di
Bologna }
 \centerline{\it Piazza di Porta S. Donato, 5. 40126 Bologna, Italy}
 \centerline{\it and}
  \centerline{\it INFN, Sezione di Bologna, Via Irnerio 46, I-40126, Bologna, Italy.}

\centerline{{\footnotesize e-mail: emanuele.latini@UniBo.it}}

\bigskip

\centerline{ M. A. Lled\'{o} }

\smallskip

 \centerline{\it  Departament de F\'{\i}sica Te\`{o}rica,
Universitat de Val\`{e}ncia and}
 \centerline{\it IFIC (CSIC-UVEG)}
 \centerline{\small\it C/Dr.
Moliner, 50, E-46100 Burjassot (Val\`{e}ncia), Spain.}
 \centerline{{\footnotesize e-mail: maria.lledo@ific.uv.es}}

\bigskip

\centerline{ F. A. Nadal }

\smallskip

 \centerline{\it  Departament de F\'{\i}sica Te\`{o}rica,
Universitat de Val\`{e}ncia and}
 \centerline{\it IFIC (CSIC-UVEG)}
 \centerline{\small\it C/Dr.
Moliner, 50, E-46100 Burjassot (Val\`{e}ncia), Spain.}
 \centerline{{\footnotesize e-mail: felip.nadal@gmail.com}}

\vskip 1cm

\begin{abstract}

In this paper study the quantum deformation of the superflag $\Fl(2|0, 2|1,4|1)$, and its big cell,  describing the complex conformal and Minkowski superspaces respectively. In particular, we realize their projective embedding via a generalization to the super world of the Segre map and we use it to construct a quantum deformation of the super line bundle realizing this embedding. This strategy allows us to obtain a description of the quantum coordinate superring of the superflag that is then naturally equipped  with a coaction of the quantum complex conformal supergroup  $\rSL_q(4|1)$.

\end{abstract}


\section{Introduction}\label{intro-sec}
The  construction of the flag and Grassmannian supermanifolds from an algebraic geometric point of view appeared for the first time
in Ref. \cite{ma1}. In a series of more recent works \cite{flv, cfl, cfl2, fl}, the superflag $\F:=\mathrm{Fl}(2|0,2|1;4|1)$ is studied at the classical and quantum level  as an homogeneous superspace for the complex special linear supergroup $\rSL(4|1)$,  the complex superconformal group of spacetime in dimension four. In the complex case we do not need to worry about the signature, but we warn the reader that the real form of the conformal group\footnote{Notice that $\rSU(2,2)$ is the  spin group of the conformal group $\rSO(4,2)$.} corresponding to a conformal structure of signature $(+,-,-,-)$ is $\rSU(2,2|1)$  and that the real form of the homogeneous space {\sl is not}  a real flag supermanifold. In both cases, the big cell of the real or complex supermanifold can be interpreted as the Minkowski superspace, once one realizes that the
subgroup of the conformal group that leaves invariant the big cell is the Poincar\'{e} supergroup times dilations. In Refs.  \cite{cfl, cfl2, fl}, a quantization of the whole scheme is shown that starts by substituting the group $\rSL(4|1)$ by its quantum counterpart, the quantum supergroup $\rSL_q(4|1)$ \cite{ma2}, and continues by constructing the corresponding quantum homogeneous superspaces. The superspace $\C^{4|1}$ that appears in $\F$ is the  {\it super twistor space}, similar to the {\it twistor space} of Penrose \cite{pe} and shall not be confused with the super spacetime.

The quantization of Grassmannians and flag manifolds (non super case)  has been studied before in several approaches. The procedure to quantize $\F$ used in this paper and in Refs. \cite{cfl, cfl2, fl} is inspired  by the approach of  Refs. \cite{fi2, fi1, fi3}  for the non super case.

There are   other approaches to the quantization of Grassmannians and flags in the literature. We  mention very briefly some of them, although our list is probably not complete. In Ref. \cite{ch} one finds an interesting relation between twistors and geometric quantization. In Refs. \cite{han,kko} the conformal group is deformed in the $R$-matrix approach. Since flags and Grassmannians are coadjoint orbits of the group $\rSL(n)$, one has the {\it Kirillov-Kostant-Soriau symplectic form} on them and they can be quantized in terms of star products and Moyal brackets \cite{fl4, ll, fll, fl5}. In Refs. \cite{al, eem, mu} the property of being coadjoint orbits is also exploited using the so-called {\it Shapovalov pairing} of Verma modules. Finally we can mention the possibility of quantizing these spaces as {\it fuzzy spaces} \cite{dmu, dj}.

\medskip

The beauty of our approach is that the actions of the (classic and quantum) conformal and
Poincar\'{e} supergroups on the  superspaces  are built into the very definition
of the superspaces. Moreover it naturally leads to a differential star product \cite{cfln,cfln2}.

It is  then  a  very  natural approach, giving the importance that symmetries have in physics.   Quantum geometry is in fact a  fascinating subject that attracted a lot of attention in the recent past, but that seems not to have exhausted its full potential: it generalizes the concepts of geometry \cite{co, gvf} in a way resemblant to what quantum mechanics does for classic mechanics and it is plausible that it may be used to describe the physical phenomena at very short distance scales, where quantum effects affecting spacetime should be noticeable. Although mainly inspired by fundamental principles, it has also appeared repeatedly in string theory and supergravity \cite{cds, sw, se, fel}. A warning, however: here we deal only with flat spaces, so this is to be considered a toy model which does not contain  gravity. But it is, in any case, a non trivial, first step towards that direction.
We also want to remark that conformal geometry have been recently applied to the study of the AdS/CFT correspondence \cite{glw}, thus we believe that our approach could also be relevant in order to consider a deformed version of such correspondence.

\medskip

Even when the full quantization of $\F$ in this scheme is proven to exist, it is extremely difficult to give a presentation of the non commutative ring associated to it in terms of generators and relations. It is possible, although not straightforward \cite{cfl, cfl2, fl}, to give such presentation for what is called in physics the {\it chiral} and {\it antichiral} superspaces. At the classical level, these spaces   do not have a real form compatible with the action of  the conformal supergroup $\rSU(2,2|1)$. They are, though, used in physics (for example in supersymmetric Yang-Mills theories \cite{fz}).

Geometrically, the chiral conformal space is the super Grassmannian $\Gr_2:=\Gr(2|1; 4|1)$ and  the antichiral is $\Gr_1:=\Gr(2|0; 4|1)$. One Grassmannian is dual to the other in a sense that will be explained in detail in Section \ref{superflag-sec}.

All the results are  in principle valid for a spacetime of dimension 4, but some aspects could be generalized to higher dimensions. For example,  in Ref. \cite{flat}   a symplectic realization of the chiral conformal superspace is proposed,
which can  be extended to the 6 and 10 dimensional cases by using matrix groups over quaternions and octonions. Also, a generalization to split signatures $(n,n)$ has been considered recently in Refs. \cite{flm,flm2}.

\medskip

The key ingredient of our treatment is to consider the projective embedding of $\F$ into a suitable projective superspace. This is achieved by viewing $\F$ as embedded inside the product of two super Grassmanians, each of them having a super Pl\"{u}cker embedding in a projective space. We give and explicit presentation of the coordinate ring with respect to such  embedding: this will be achieved by using two sets of Pl\"{u}cker coordinates, constrained by a set of relations called {\it incidence relations} which are, in some sense,  orthogonality conditions. The big cell of $\F$  will turn out to be the complex Minkowski superspace as constructed in Ref.
\cite{flv}. Moreover, we use the projective localization technique to give a local picture of the conformal superspace: the incidence relations become locally the {\it twistor relations} used in physics.

The coordinate rings of $\Gr_1$, $\Gr_2$ and $\F$ can be seen as subrings of coordinate ring of $\rSL(4|1)$. We exploit this feature to tackle the problem of their quantization. We consider the quantum supergroup $\rSL_q(4|1)$ and then we try to identify subalgebras of it that can be used to define the quantum super Grassmannians and the quantum superflag. While for the Grassmannians this technique is successful (although laborious), for the superflag is much more involved: finding the commutation relations between the  coordinates of $\Gr_1$ and $\Gr_2$ seems an unsurmountable calculation. On top of it, one should find also the quantum version of the incidence relations.

We then try a different strategy. It is well know that the product of two projective spaces can be further embedded into a bigger projective space through the Segre embedding. This can be generalized to the super setting, and then one can use the relation between projective embeddings and very ample line bundles, which also holds in the super case.

Let $G$ be an algebraic Lie supergroup and $P$ a parabolic subgroup and consider the homogeneous space $G/P$ (in our case, $G=\rSL(4|1)$ and $P$ is either of the parabolic subgroups associated to $\Gr_1$, $\Gr_2$ and $F$).  If $\cO(G)$ and $\cO(P)$ are the coordinate superrings of $G$ and $P$. Then such embeddings can be given in terms of an element of $\cO(G)$ lifted from a character of $P$ \cite{fi6}. We call this the {\it classical section} of the embedding. By the method of parabolic induction one can see that equivariant sections of $\cO(G)$ with respect to the $n^{\mathrm{th}}$ power of such character give the degree $n$ subspace of the graded coordinate ring $\cO(G/P)$ associated to the projective embedding of $G/P$.

Luckily, it is possible to translate this approach to the quantum realm by means of a {\it quantum section} \cite{fi6,CFG,FG}. We then achieve a characterization of the coordinate superring of $\F$ associated to the super Segre embedding.

\medskip

We have tried to give at least an idea or sketch of the different notions that appear in the paper. The reader that is versed in algebraic geometry may skip those parts, but a mathematical physicist may find this a useful guide, even when proofs of standard results could not be provided explicitly. We give a list of references for all the details.

\medskip

The organization of the paper is as follows:

In Section \ref{algebraicsupervarieties-sec} we  give a  self contained introduction to supergeometry. We use it to set up the notation and clarify the language used in the paper. In particular, we describe  the functor of points approach to supergeometry, explaining some fundamental examples such as the projective superspace and algebraic supergroups.

In Section \ref{superflag} we  study the embedding of $\F$ in $\Gr_1\times \Gr_2$, and then the super Pl\"{u}cker embeddings of $\Gr_1$ and $\Gr_2$ into projective superspaces. We give  the superring characterizing such embedding for the flag, including the incidence relations.  Then we review the Segre map and we propose its supersymmetric generalization.

In Section \ref{linebundles-sec} we give a detailed description of the parabolic (super) geometries that enter into play. Then we describe the very ample super line bundles associated to the Pl\"{u}cker embeddings of $\Gr_1$ and $\Gr_2$, the  {\it bundles of antichiral}  and {\it chiral superconformal densities}, respectively. These are the building blocks  of the construction of the very ample super line bundle for the superflag, whose description  will be postponed until Section \ref{qsection-sec} for technical reasons.

In Section \ref{qgrass-sec} we recall briefly the definition of the quantum group $\rSL_q(4|1)$ and we propose a natural deformation of $\Gr_1$ and $\Gr_2$ as subalgebras of $\rSL_q(4|1)$ .

 In Section \ref{qsection-sec} we construct the very ample line bundle realizing the Segre embedding of $\F$ and we use the notion of quantum section to propose a characterization of the quantum coordinate superring for the super flag.

 Finally, in the Appendix \ref{incidence-ap} we write down the incidence relations explicitly and in Appendix \ref{embedding-ap} we prove that the super Segre map is an embedding using the so-called {\it even rules principle}.

  \medskip

{\bf Notation.} The reader may find useful to resort to Refs. \cite{dm,va,ccf} for all the results in supergeometry and to Ref. \cite{fl} for a more specific treatment of conformal and Minkowski superspaces.
\section{Algebraic supervarieties}\label{algebraicsupervarieties-sec}

In this section we intend to give a very brief account of some notions of supergeometry that are an extension to superalgebras of well known notions of algebraic geometry. We will assume that concepts as the spectrum of a ring, sheaf, ringed space, Zariski topology, etc are known and we will try to sketch how one proceeds to the generalization. We will not write any proof of the statements made, that can be found in the above mentioned references.  We set the ground field to be $k=\R$ or $\C$.

\begin{definition}  A commutative\footnote{We try to avoid the use of `supercommutative algebra' as it appears in the physics literature and stick to the categorical definition \cite{dm}.} superalgebra $A=A_0\oplus A_1$ is said to be an {\it affine superalgebra} if its even part, $A_0$, is finitely generated as an algebra, its odd part, $A_1$, is finitely generated as an $A_0$-module and the reduced algebra, $A_{0,r} =A_0/J_{A_0}$, where $J_{A_0}$ is the ideal of odd nilpotents, is itself affine (contains no nilpotents).

 \hfill$\square$
\end{definition}

 In ordinary geometry, for any commutative, affine algebra $F$ one finds an affine algebraic variety as the {\it spectrum} of $F$. The topological space   $|X|= \spec(F)$ is the set of prime ideals of $F$ endowed with the {\it Zariski topology}.  The  {\it structure sheaf} is constructed by localizing $F$ at each each $\fp\in \spec(F)$. The algebra
$${F}_\fp:=\left\{\frac f g\;\; \big|\;\; f\in F,\; g\in  F- \fp\right\}$$
 is the stalk of the sheaf at $\fp$. The structure sheaf is  denoted as $\cO_{X}$ and the pair  $X=(|X|, \cO_X)$ is an affine algebraic variety.
We recover the affine algebra $F$ as the set of global sections:
$$\cO(X):=\cO_{X}(X)=F\,.$$ $\cO(X)$  is  the {\it coordinate ring} or {\it coordinate algebra} of the affine variety $X$.

There is then an equivalence between the categories of affine algebras and affine algebraic varieties: it is given by the contravariant functor
$$\begin{CD} \varaff@=\algaff\\
X@>>>\cO(X)\\\spec(F)@<<<F\,.\end{CD}$$ For the morphisms, one can prove that
$$\Hom_\varaff(X, Y)= \Hom_\algaff(\cO(Y), \cO(X))\,.$$

A generic algebraic variety is constructed by gluing together affine algebraic varieties. But in order to extend the above correspondence to  general commutative algebras one has to consider a wider category,   the category of {\it affine schemes}.

\begin{definition}\label{ringed space} A {\it  ringed space} is a pair $M=(|M|, \cF)$ consisting of a topological space $|M|$ and a sheaf of commutative rings $\cF$ on $|M|$. If the stalk $\cF_x$ at each point $x\in |M|$ is a local ring (it has a unique maximal ideal) then we say that it is a {\it locally ringed space}.

A morphism of ringed spaces $\phi:M=(|M|, \cF)\rightarrow N=(|M|, \cG)$ consists of a continuous map $|\phi|:|M|\rightarrow |N|$ and a sheaf morphism $\phi^*:\cG\rightarrow |\phi|_*\cF$, where $|\phi|_*\cF$ is the sheaf on $|N|$ defined as $|\phi|_*\cF(U)=\cF(\phi^{-1}(U))$ for all open sets $U\subset |N|$. If the spaces are locally ringed then $\phi^*$, on the stalks, must sent the maximal ideals to the maximal ideals.

\hfill$\square$
\end{definition}

The pair $X=(\spec(F), \cO_X)$  constructed above is a locally ringed space. We will denote it as  $X=\uspec(F)=(\spec(F), \cO_X)$, to distinguish it from the topological space $|X|=\spec(F)$.

\begin{definition} An {\it  affine scheme} is a locally ringed space which is isomorphic to $\uspec(F)$ for some algebra $F$, not necessarily affine.

A morphism of affine schemes is a morphism of locally ringed spaces.

\hfill$\square$
\end{definition}

One can also prove that the categories of commutative algebras $\alg$ and  affine schemes $\schemesaff$ are contravariantly equivalent. Given this equivalence, we will denote indistinctly as $\cO_F$ or $\cO_X$ the structure sheaf of $X=\uspec(F)$.

\begin{definition} A {\it scheme} is a locally ringed space which is locally isomorphic to an affine scheme.

\hfill$\square$
\end{definition}

\bigskip

For an affine superalgebra $A$, the reduced algebra $A_{0,r}$ is an affine algebra, so we have an affine variety associated to it. Notice that, as topological spaces,  $\spec (A_{0,r})=\spec (A_0)$, since they differ only by nilpotents. On the other hand,  $A$ is an $A_0$-module and one can define a sheaf of $\cO_{A_0}$-modules over  $\spec (A_0)$ such that the set of its global sections is $A_1$ and the stalk at each prime $\fp\in \spec (A_0)$ is the localization of $A_1$ at $\fp$. We have then a sheaf of superalgebras over $\spec (A_0)$. This is the basis for the definition of algebraic supervariety.

\begin{definition} A {\it superspace}\footnote{One shall not mistake this general definition of superspace with the more restricted notion used in physics designating certain super spacetimes.} is a pair $S=(|S|, \cO_S)$ where $|S|$ is a topological space and $\cO_S$ is a sheaf of superalgebras whose stalk  at a point is a local superalgebra (it has a unique, two-sided, maximal ideal).

A morphism of superspaces $f: S=(|S|, \cO_S)\rightarrow T=(|T|, \cO_T)$ is a pair $f=(|f|, f^*)$ where $|f|:|S|\rightarrow |T|$ is a continuous map and $f^*:\cO_T\rightarrow |f|_*\cO_S$ a morphism of sheaves (as in Definition \ref{ringed space}).

\hfill$\square$
\end{definition}

\begin{definition} An  {\it affine algebraic supervariety} is a superspace $S=(|S|, \cO_S)$ constructed from an affine superalgebra $A=A_0+A_1$ by the procedure specified above. We will denote it as $S=\uspec(A)$.

 $|S|=\spec (A_0)$ is  an ordinary affine algebraic variety, the {\it reduced} variety.

The  affine superalgebra $A=\cO_S(S)$ is the {\it coordinate superalgebra} or {\it coordinate superring} of $S$.

Morphisms of affine supervarieties are morphisms of superspaces.

 \hfill$\square$
\end{definition}

\begin{example}\label{affine superspace} The {\it affine superspace $k^{p|q}$} is the affine supervariety  whose coordinate superalgebra is the affine superalgebra
$$A= k[x_1,\dots ,x_p]\otimes \Lambda (\theta_1,\dots \theta_q)\,,$$
where $\Lambda (\theta_1,\dots, \theta_q)$ is the exterior algebra on the indeterminates $\theta_1,\dots ,\theta_q$. The reduced algebra is the polynomial algebra $A_{0,r}=k[x_1,\dots ,x_p]$ and the reduced variety is simply $k^p$.

 \hfill$\square$
\end{example}

 By this construction, the categories of affine superalgebras $\salgaff$ and affine supervarieties $\svaraff$ are equivalent, as in the ordinary setting. In fact, one has that for two affine supervarieties
 $$\Hom (T,S)=\Hom (\cO(S), \cO(T))\,.$$

Generic algebraic varieties are superspaces that are locally isomorphic to affine supervarietes. We will encounter algebraic varieties that are not affine in the following.
As in the non super case, in order to obtain an equivalence of categories with general superalgebras we have to use affine superschemes.

\begin{definition}
An {\it affine superscheme } is a superspace $S=(|S|, \cO_S)$ which is isomorphic to $\uspec(A)$ for a superalgebra $A$, not necessarily affine.

A morphism of superschemes is a morphism of the corresponding superspaces.

 \hfill$\square$
\end{definition}

We have that the categories of affine superschemes $\sschemeaff$ and the category of commutative superalgebras  $\salg$ are contravariantly equivalent. More generally, we have de following definition

\begin{definition}\label{superscheme-def} A {\it superscheme} is a superspace $S=(|S|, \cO_S)$ such that $\cO_{S,1}$ (in $\cO_{S}=\cO_{S,0}+\cO_{S,1}$) is a quasi coherent sheaf of $\cO_{S,0}$-modules.

One can prove that a generic superscheme is locally isomorphic to $\uspec(A)$ for some superalgebra $A$.

 \hfill$\square$
\end{definition}

\bigskip

We will use the {\it functor of points} approach to supervarieties, so we recall its general definition here.  We introduce it in the general language of superschemes,  although one could also define it on algebraic supervarieties.

Let $S$ and $T$ be be superschemes.
  A {\it $T$-point of $S$} is a morphism of superschemes $T\rightarrow S$. We denote by $h_S(T)=\Hom(T,S)$ the set of all $T$-points of $S$.

\begin{definition}The functor of points of a superscheme $S$ is the contravariant  functor $h_S: \sschemes\rightarrow \sets$ defined on objects as
$$T\rightarrow h_S(T):=\Hom_{\sschemes}(T,S)\,,$$
and on morphisms $\phi:T\rightarrow T'$ as
$$h_S(\phi)f=f\circ \phi,\qquad f\in \Hom_{\sschemes}(T', S)\,.$$
\hfill$\square$
\end{definition}

One can  prove that morphisms of supervarieties or superschemes are natural transformations between their functors of points.

The following important result will be used repeatedly in applications.

\begin{theorem}\label{theorem restriction}
The functor of points of a superscheme is determined by its restriction to the category of affine superschemes.

\hfill$\square$
\end{theorem}

So, once we have the functor defined in general,  for many purposes it will be enough to check how it works in the 
subcategory
\be \begin{CD}\sschemeaff @>h_S>>\sets \\
T@>>>\Hom(T, S)\,.\end{CD}\label{functorrepr}\ee

A  contravariant functor $h$ from a certain category $(\mathrm{cat})$ to $\sets$ is said to be {\it representable} if $h(X)=\Hom (X, Y)$ for all $X \in (\mathrm{cat})$ and some $Y \in (\mathrm{cat})$. The object $Y$  is said to {\it represent} the functor $h$. It is customary to denote the representable functor as $h_Y$. So the functor of points of a superscheme is a contravariant, representable functor. For a covariant functor we just have to substitute $h(X)=\Hom (Y, X)$.

Notice then  that the functor defined in (\ref{functorrepr}) will not be representable in the category of affine superschemes if $S$ is not affine. Nevertheless, given the equivalence of categories among $\sschemeaff$ and $\salg$, we can give alternatively the restricted functor (\ref{functorrepr}) as
$$\begin{CD}\salg @>h_S>>\sets \\
A@>>>\Hom(\uspec(A), S).\end{CD}$$
If the superscheme is affine, then  $\Hom(\uspec(A), S)=\Hom(\cO(S), A)$.

In the literature, sometimes these functors are also called `representable', but we have to be aware that, strictly speaking, they are not: the functor of points of a generic superscheme is always representable in the category of superschemes.
If the superscheme is affine, then for all purposes it is enough  to define the functor of points on affine superschemes. The extension to generic superschemes could be done by a {\it gluing procedure}, using the fact  mentioned in Definition \ref{superscheme-def}.

\begin{example}\label{fop affine superspace}{\it The functor of points of affine superspace.} We now go back to Example \ref{affine superspace}. Let $k^{p|q}$ be the scheme of the affine superspace. We want to describe its functor of points on affine superschemes.    Due to the equivalence of categories, we have that a  morphism
$$\phi:T\rightarrow k^{p|q},\qquad T\in \sschemeaff$$
can be determined by a morphism
$$\begin{CD} A=k[x_1,\dots ,x_p]\otimes \Lambda (\theta_1,\dots, \theta_q)@>\phi'>> B=\cO(T)\\
(x_1,\dots ,x_p\,;\,\theta_1,\dots, \theta_q)@>>>(a_1,\dots ,a_p\,;\,\alpha_1,\dots, \alpha_q)\,.\end{CD}$$
The morphism $\phi'$ is determined by the image of the generators, that is, by $p$ even elements $a_i\in B_0$ and $q$ odd elements $\alpha_j\in B_1$.

\hfill$\square$
\end{example}

Notice that taking $B= k$, the $k$-points of of $A$ are just the geometric points of the affine space $k^p$, the reduced variety. To detect the presence of odd variables one needs more than the geometric points, one needs  the full functor of points.

By an abuse of notation, we will use the same name for the functor given on the affine superschemes as for the functor given on the superalgebras. Moreover, when there is no possibility of confusion, we will denote these functors with the same letter than the superscheme or the supervariety itself.

The following example of affine supervariety will be of special interest for us.

\begin{example}\label{algebraicsupergroup}
{\it The functor of points of the algebraic supergroup $\rSL(4|1)$.} An algebraic supergroup is an algebraic supervariety whose functor of points is group valued.

Let $A$ be a superalgebra.  A supermatrix  of dimension $m|n$ with entries in $A$ is a matrix of the form

$$M_{n|m}(A):=\left\{\begin{pmatrix}p_{m\times m}&q_{m\times n}\\r_{n\times m}&s_{n\times n}\end{pmatrix}\right\}$$
where $p, q, r$ and $s$ are blocks of the dimension indicated above, with the entries of $p$ and $s$ are valued in $A_0$ and the entries of $q$ and $r$ are valued in $A_1$.
Let $A$ be a superalgebra.  We define the functor $\rSL(4|1): \salg\rightarrow \sets$ as
$$\rSL(4|1)(A):=\left\{g=\begin{pmatrix}p_{4\times 4}&q_{4\times 1}\\r_{1\times 4}&s_{1\times 1}\end{pmatrix}\in M_{4|1}(A);\, \, p, s \hbox{ invertible, }\Ber\, g=1\right\}\,.$$
The Berezinian is given by (we suppress the dimensions of the blocks)
$$\Ber\begin{pmatrix}p&q\\r&s\end{pmatrix}=\det s^{-1}\cdot\det(p-qs^{-1}r)\,.$$ The Berezinian is the generalization of the determinant to the super case. Notice that it is only defined if $s$ and $p$ are invertible, which in turn implies that the supermatrix is invertible.

On morphisms $f:A\rightarrow B$, the functor $\rSL(4|1)$ behaves as follows.
Let us denote as $g=(g_{ij}) \in \rSL(4|1)(A)$, the supermatrix formed with the entries $g_{ij}$. Then
$\rSL(4|1)(f)(g)=(f(g_{ij})) \in \rSL(4|1)(B)$.
 This can be defined for the full $M_{4|1}$ (which is also a functor) and then one can check that  it preserves the invertibility condition.

This functor is group valued and representable. It corresponds to an algebraic  supervariety with (affine) coordinate superalgebra
$$k[x_{ij},\xi_{kl}][(\det p)^{-1}, (\det s)^{-1}]/( \Ber \, g -1)\,,$$ where
$x_{ij}$ are the even entries of $p$ and $s$, and $\xi_{kl}$ are the odd entries of $q$ and $r$.

\hfill$\square$
\end{example}

Examples of schemes and superschemes that are not affine are the projective space and superspace and projective algebraic varieties and supervarieties. We are interested in some examples of projective supervarieties   so, as a first example, we will describe the projective superspace.

\begin{example}\label{fop projective superspace}{\it Functor of points of the projective superspace.} Let $T=(|T|,\cO_T)$ be an affine superscheme. We define the following functor

\be \begin{CD}\sschemes @> \bP^{m|n}>> \sets\\
T@>>>\bP^{m|n}(T)\,,\end{CD}\label{functor projective space}
\ee
where
$$\bP^{m|n}(T):=\{\,\hbox{locally free subsheaves }  \cF_T\subset k^{m+1|n}\otimes \cO_T \hbox{ of rank }1|0\,\}\,.$$
When the functor is restricted to  $\sschemeaff$ one can  give it equivalently on  $\salg$:

\be \begin{CD}\salg @> \bP^{m|n}>> \sets\\
A@>>>\bP^{m|n}(A)\,,\end{CD}\label{functor projective space alg}
\ee
where
$$\bP^{m|n}(A):=\{\hbox{\small{finitely generated projective submodules} }  M\subset A^{m+1|n} \hbox{ \small {of rank} }1|0\}$$
and $A^{m+1|n}:=A\otimes k^{m+1|n}$. On morphisms $f:A\rightarrow B$ we have that $\bP^{m|n}(f)$ is given through the extension of scalars
$$\begin{CD}
\bP^{m|n}(A)@>\bP^{m|n}(f)>>\bP^{m|n}(B)\\
M_A@>>>B\otimes_A M_A\,.
\end{CD}$$

This definition is closer to the geometric interpretation. One needs to prove that the functor is representable in the  $\sschemes$ category, that is, we have to find a superscheme $\bP^{m|n}$ (by an abuse of notation, we will denote the functor and the superscheme by the same symbol), such that
$$\bP^{m|n}(T)=\Hom (T,\bP^{m|n})\,.$$
The superscheme $\bP^{m|n}$ is in fact an algebraic variety, but not an affine superscheme. It is constructed as a sheaf over the  ordinary projective space $\bP^{m}=\C^{m+1}/\sim$, where, as usual
$$(x_0, x_1,\dots , x_m)\sim (x'_0, x'_1,\dots , x'_m)\qquad \hbox{iff} \quad x'_i=\lambda x_i, \lambda\in \C^\times\,.$$
An open cover of $\bP^{m}$ is given by
$$U_i=\{[x_0, x_1,\dots, x_m] \; |\; x_i\neq 0\}  \qquad i=1,\dots, m\,,$$ and local coordinates are given by
$$(u^i_0,\dots, {\hat{u}}^i_i,\dots , u^i_m),\qquad u^i_k:=x_k/x^i\,.$$
 On each open set $U_i$ we consider the sheaf that for any $V$, open subset of $U_i$,
 $$\cO_{U_i}(V)=\cO_{0,\,U^i}(V)\otimes \wedge (\xi_1,\dots, \xi_n)\,,$$ where $\cO_{0,\,U^i}(V)$ is simply the sheaf of algebraic functions on $U_i$ and then
 so $\cO_{U_i}\cong \C^{m|n}$.

 These sheafs can be glued conveniently, the procedure being similar to the one used for standard projective space. In this way we obtain the supervariety $ \bP^{m|n}$ representing the functor (\ref{functor projective space}). For a detailed proof of this fact, see Refs. \cite{ccf, cfl, fl}.

\hfill$\square$
\end{example}

\bigskip

A projective supervariety is a supervariety that can be embedded into a projective superspace. There is a standard construction, parallel to that of $\uspec(F)$, that contains the information of the projective supervariety and the embedding. One starts with a $\Z$-graded superring (with grading compatible with the $\Z_2$-grading) and considers {\it homogeneous ideals}. We describe briefly  the construction, starting with the projective space itself.

\begin{example} {\sl Projective superspace and projective supervarieties.} \label{constructionproj} The standard construction of the projective superspace as a superscheme follows the same lines than for the non super case (see for example Chapters 2 and  10 in Ref. \cite{ccf}).
We first consider  the $\Z$-graded superalgebra $A=k[x_0,x_1,\dots,x_m ;\xi_1,\dots,\xi_n]\,,$ with the polynomial grading, which is compatible with the $\Z_2$-grading. The polynomial grading is denoted with superindices , $A=\oplus_{n=0}^\infty A^n$, while the $\Z_2$-grading is denoted with subindices $A=A_0+A_1$. The topological space $\proj(A_0)$ is the set of $\Z$-homogeneous prime ideals in $A_0$ which do not contain the ideal $(x_0,x_1,\dots,x_m ;\xi_i\xi_j)$. This set is given the Zariski topology. It can be covered by open sets $U_i$, $i=0, \dots ,n $, formed by the homogeneous primes not containing $x_i$, which amounts to say, the homogeneous primes in $$k[x_0,x_1,\dots,x_m ;\xi_1,\dots,\xi_n][x_i^{-1}]\,,$$ or simply the prime ideals in the ring at degree zero  $$\left(k[x_0,x_1,\dots,x_m ;\xi_1,\dots,\xi_n][x_i^{-1}]\right)^0\,.$$ Setting $u_j=x_j/x_i$, this ring is isomorphic to
$$k[u_0,\dots,\hat{u}_i,\dots ,u_m ;\xi_1,\dots,\xi_n]\,,$$ which defines the sheaves $\cO_{U_i}$ of  Example \ref{fop projective superspace}. They glue to a scheme that we denote as  $\uproj(A):=(\proj(A), \cO_A)=\bP^{m|n}$.

The stalk of the scheme at each prime $p$ is obtained by localizing $A$ as an $A_0$-module, essentially with the technique described in Ref. \cite{ha} page 116.
\smallskip

A projective algebraic supervariety can be constructed  in a similar way starting from a graded superring $S=\oplus_{d=0}^\infty S^d$ obtained by quotienting $A$ by some homogeneous, $\Z_2$-graded ideal.   As before, the set $\proj(S_0)$ is the set of all homogeneous prime ideals not containing $S_0^+=\oplus_{d=1}^\infty S^d_0$. This set is given the Zariski topology and a sheaf $\cO_S$ obtained localizing $S$. We then obtain a scheme  denoted as $\uproj(S):=(\proj(S), \cO_S)$, locally isomorphic to an affine scheme.


\hfill$\square$
\end{example}

\begin{remark} \label{duals-rem}{\it Duals and inner homomorphisms for modules.} Let $A$ be a commutative superalgebra (associative, with unit). For a commutative superalgebra, a left $A$-module $M$ is also a right $M$-module by setting
 $$m\cdot a:= (-1)^{p(m)p(a)}a\cdot m,\qquad a\in A, m\in M\,,$$ where $p(a)$ and $ p(m)$ are the parities of $a$ and $m$,
 so we can call them just  `modules'. Let $\cM$ be the category of $A$-modules. It is a tensor category
 with unit $1=A$. Morphisms in $\cM$ are linear maps of $A$ modules $f:M\rightarrow M'$
that preserve parity.

If we consider the free modules $M=A^{p|q}$ and $M=A^{r|s}$, a morphism $f$ can be represented as a supermatrix
\be \begin{pmatrix}T_{r\times p}&T_{r\times q}\\T_{s\times p}&T_{s\times q}\end{pmatrix}\label{morphismfree}\ee
where the diagonal blocks have entries in $A_0$ and the off-diagonal blocks have entries in $A_1$.

We can consider linear maps $N\rightarrow M$, not necessarily parity preserving. For free modules, this will correspond to consider supermatrices as in (\ref{morphismfree}), but with arbitrary entries in $A$. We will denote them as $\rhom(N, M)$ as opposed to $\Hom(N, M)$ for the parity-preserving morphisms. It is clear that $\rhom(N, M)$ is an $A$-module itself. One says that $\rhom(-,-)$ is an {\it internal Hom functor}
$\cM^{\mathrm{op}}\times \cM\rightarrow \cM$.

The {\it dual} of an $A$-module $M$ is the $A$-module
$M^*=\rhom(M, A)$. Let us consider the free module $M=A^{p|q}$ and let  $\{e_1, \dots, e_p, \,\ep_1,\dots \ep_q\}$ denote the canonical basis\footnote{Latin letters will generically denote even quantities, Greek letters odd ones.}. As for vector spaces one can define the {\it dual basis} $\{e_1^*, \dots, e_p^*, \,\ep_1^*,\dots, \ep_q^*\}$ with
\be e_i^*(e_j)=\delta_{ij},\quad  e_i^*(\ep_\alpha)=0, \quad \ep_\alpha^*(e_i)=0,\quad \ep_\alpha^*(\ep_\beta)=\delta_{\alpha\beta},\label{dualbasis}\ee
with $i,j=1,\dots,p\,$ and $\,\alpha,\beta=1,\dots,q$. Any linear map $M\rightarrow A$ is determined by its values on a basis so it will be an $A$-linear combination of the dual basis. The relations (\ref{dualbasis}) imply that $e_i^*$ is a parity preserving (even)  map while $\ep_\alpha^*$ changes by one the parity (odd map), so  $M^*\cong A^{p|q}$, although the isomorphism is not natural.

We will be interested in finitely generated projective modules and their duals. One has that given two such $A$-modules $N$ and $M$
$$(N\oplus M)^*=N^*\oplus M^*\,.$$ If $M$ is a projective module, the above equation implies that $M^*$ is also a projective module. One also has the natural isomorphism
$${M^*}^*\cong M\,.$$

\hfill$\square$

\end{remark}

\section{The classical
superflag $\Fl(2|0,2|1;4|1)$ and its projective embeddings}

\label{superflag}The super flag manifold $\Fl(2|0,2|1;4|1)$ of $2|0$-subspaces of $2|1$-subspaces of
the complex superspace $\C^{4|1}$ is the model of the complex, flat\footnote{We refer to flat conformal geometry when the  Minkowski space is contained  in the variety that we consider.}, conformal, $N=1$ superspace. Indeed, this supervariety has a real form that is an homogeneous space for the  the conformal supergroup $\rSU(2,2|1)$ and contains the super Minkowski space as its big cell, together with the appropriate action of the super Poincar\'{e} subgroup on it. These are all physical requirements.   In Section \ref{homogeneous} we make a brief summary of the properties of the conformal and Minkowski superspaces seen as homogeneous superspaces of the conformal and Poincar\'{e} supergroups, respectively.

We will realize the superflag $\Fl(2|0,2|1;4|1)$ as contained in the product of $\Gr(2|0;4|1)$ (the super Grassmannian
 of $2|0$-subspaces in $\C^{4|1}$) times its dual, $\Gr(2|1;4|1)$ (the super Grassmannian
 of $2|1$-subspaces in $\C^{4|1}$).

 \smallskip

To ease the notation,  through this paper we will denote (as in the Introduction) $$ \F:=\Fl(2|0,2|1;4|1), \qquad \Gr_1:=\Gr(2|0;4|1),\qquad \Gr_2:=\Gr(2|1;4|1)\,.$$

\smallskip

 In the super context,
not all Grassmannian and flag superschemes can be embedded into a projective
space: topological obstructions arise some cases. An example of a non projective super Grassmannian is $\mathrm{Gr(1|1,2|2)}$ (see Chapter 4 in Ref. \cite{ma1} and Chapter 10 in Ref. \cite{ccf}). But in the special case that we investigate here, it is possible  to obtain a projective embedding by generalizing to the super case the so called  Pl\"{u}cker embedding of the Grassmannian $\Gr(2,4)$ in $\bP(\wedge ^2\C^4)\cong \bP(\C^6)\cong \bP^{5}$.

Our presentation here goes along the lines explored in  Refs. \cite{ma1, flv, cfl, fl}. We will give first the super Pl\"{u}cker embedding of $\Gr_1$ and $\Gr_2$ in the projective space $\bP^{6|4}$. Since $\F\subset \Gr_1\times \Gr_2$, then it also admits a projective embedding in $\bP^{6|4}\times \bP^{6|4}$. Finally, we will describe the Segre embedding of $\bP^{6|4}\times \bP^{6|4}$ in the super projective space $\bP^{64|56}$, which in turns gives the projective embedding of $\Gr_1\times \Gr_2$ and $\F$ as  subsets of $\bP^{64|56}$.

\smallskip

We start first with the Pl\"{u}cker embedding of the super Grassmannians.

\subsection{Pl\"{u}cker embedding of the super Grassmannian \\ $\Gr(2|0;4|1)$ and its dual $\Gr(2|1;4|1)$}\label{pluckergrass}

Let $\{e_1, e_2, e_3, e_4, \ep_5\}$ be the canonical basis of $\C^{4|1}$
with $e_i$ even and $\ep_5$ odd. As customary, Latin letters are used
for even objects, while Greek letters are used for  odd ones.
Let $A$ be a commutative superalgebra. We define the functors  $\Gr_1, \Gr_2:\salg\rightarrow \sets $  as
\begin{align*}
\Gr_1(A)&:=\{\,\hbox{finitely generated, projective submodules }  M\subset A^{4|1} \hbox{ of rank }2|0\,\}, \\ \\
\Gr_2(A)&:=\{\,\hbox{finitely generated, projective submodules }  M\subset A^{4|1} \hbox{ of rank }2|1\,\}\,.
\end{align*}
On morphisms of superalgebras $f:A\rightarrow B$ we have that $\Gr_1(f)$ is given by the extension of scalars
$$\begin{CD}\Gr_1(A)@>\Gr_1(f)>> \Gr_1(B)\\
M_A@>>>B\otimes_A M_A\,.
\end{CD}$$

These definitions are analogous to the definition of projective superspace (\ref{fop projective superspace}) and the geometrical meaning is clear. As in that case, one can prove representability (in the sense of Theorem \ref{theorem restriction} and the comments following it) of these functors in terms of  superschemes over the reduced variety, the Grassmanian of 2-planes in $\C^4$, \emph{i.e.}  $\Gr(2,4)$.  Moreover, the supergroup $\rSL(4|1)$ has a left action over $\Gr_1$ and $\Gr_2$, so they become homogeneous spaces
$$\Gr_1=\rSL(4|1)/P_1,\qquad \Gr_2=\rSL(4|1)/P_2\,,$$
where $P_1$ and $P_2$ are certain parabolic subgroups of $\rSL(4|1)$. In the functor of points notation, they are explicitly
\begin{align}&P_1(A)=\left\{\begin{pmatrix}
g_{11} &g_{12}&g_{13}&g_{14}&\gamma_{15}\\
g_{21} &g_{22}&g_{23}&g_{24}&\gamma_{25}\\
0 &0&g_{33}&g_{34}&\gamma_{35}\\
0 &0&g_{43}&g_{44}&\gamma_{45}\\
0 &0&\gamma_{53}&\gamma_{54}&g_{55}\end{pmatrix}\right\}\,,\nonumber\\\nonumber\\
&P_2(A)=\left\{\begin{pmatrix}
g_{11} &g_{12}&g_{13}&g_{14}&\gamma_{15}\\
g_{21} &g_{22}&g_{23}&g_{24}&\gamma_{25}\\
0 &0&g_{33}&g_{34}&0\\
0 &0&g_{43}&g_{44}&0\\
\gamma_{51} &\gamma_{52}&\gamma_{53}&\gamma_{54}&g_{55}\end{pmatrix}\right\},\qquad g_{ij}\in A_0,\quad \gamma_{kl} \in A_1\,.\label{parabolics}\end{align}
For the proof of these facts we refer the reader to Refs. \cite{ccf, fl}. The treatment of homogeneous spaces of supergroups, also with the functor of points approach, is done carefully in Ref. \cite{ccf}.

In Ref.  \cite{cfl}
 the embedding of $\Gr_1$ in the projective space $\bP(E)$, where  $E=\wedge^2\C^{4|1}\cong\C^{7|4}$, is described by giving explicit coordinates in
the functor of points approach. In the notation of Example \ref{fop projective superspace}, $\bP(E)\cong \bP^{6|4}$. We will briefly outline it here.

In the functor of points language, any morphism (in particular, an embedding) is a natural transformation among the functors. We then need to give a natural transformation among the functors $\Gr_1,\,\bP(E):\salg\rightarrow\sets$. For each object $A$ in $\salg$  we define the morphism
$$\begin{CD}
\Gr_1(A)@>p_A>>\bP(E)(A)\\
M@>>>\wedge^2M\,,
\end{CD}$$
where
$$\wedge^2M=M\otimes M/(u\otimes v+ (-1)^{|u||v|}v\otimes u)$$ and $|u|$ is the parity of $u$.
Notice that if $M$ is a projective submodule of $A^{4|1}$ of rank $2|0$, then $\wedge^2M$ is a projective submodule of $\wedge^2A^{4|1}\cong A^{7|4}$ of rank $1|0$ so it is indeed an element of $\bP(E)(A)$.
This morphism is functorial in $A$, that is, given a superalgebra morphism $f:A\rightarrow B$, the diagram
$$\begin{CD}
\Gr_1(A)@>p_A>>\bP(E)(A)\\
@V\Gr_1(f)VV@ VV\bP(E)(f)V\\
\Gr_1(B)@>p_B>>\bP(E)(B)
\end{CD}$$
commutes. It is an easy exercise to prove that this is true.

The functors of points of $\Gr_1$ and $\bP(E)$ are {\it local} or {\it sheaf} functors (see for example Chapter VI of Ref. \cite{eh} in the non super case, Appendix B.2 of Ref. \cite{ccf} in the super case). In our case, this is guaranteed by the fact that they are functors of points of superspaces. A natural transformation between local functors is determined by its behaviour on local superalgebras This result is proven in  Proposition B.2.13 of Ref. \cite{ccf}, and it is a generalization of a similar result in the non super case (see, for example, Ref. \cite{eh}). So, once we have defined the natural transformation $p_A$ for an arbitrary superalgebra, we can restrict ourselves to work on local superalgebras.   The projective submodules over a local superalgebra are free; then they have a  basis, which considerably  simplifies the treatment.

Let $A$ be a local superalgebra and let
$W_1(A) \in \Gr_1(A)$ be the linear span over $A$ of two linearly independent, even vectors. In  the canonical basis of $\C^{4|1}$ (as above) we have
\be
W_1(A) = \rspan\left\{  \begin{pmatrix} a_{11} \\
a_{21} \\ a_{31} \\ a_{41} \\ \alpha_{51} \end{pmatrix}
\, , \,
\begin{pmatrix} a_{12} \\
a_{22} \\ a_{32} \\ a_{42} \\ \alpha_{52} \end{pmatrix} \right\}
\, = \, \rspan\left\{ r+\rho \ep_5, \, s+\sigma \ep_5 \right\} \, \subset \,
A^{4|1}\,,\label{w1a}
\ee
with

\begin{align*}
&r=a_{11}e_1+a_{21}e_2+a_{31}e_3+a_{41}e_4,\qquad  \rho=\alpha_{51}, \\
&s=a_{12}e_1+a_{22}e_2+a_{32}e_3+a_{42}e_4, \qquad  \sig=\alpha_{52}, \qquad a_{ij}\in A_0,\; \alpha_{5k}\in A_1\,.
\end{align*}
We consider now the wedge product of the two vectors. A basis in $E(A)=\wedge^2\C^{4|1}(A)$ is given by
\begin{align*}
&e_1\wedge e_2, \quad e_1\wedge e_3, \quad e_1\wedge e_4, \quad e_2\wedge e_3, \quad e_2\wedge e_4, \quad e_3\wedge e_4, \quad \ep_5\wedge \ep_5\quad \hbox {(even)},\\
&e_1\wedge \ep_5,\quad e_2\wedge \ep_5,\quad e_3\wedge \ep_5,\quad e_4\wedge \ep_5,\quad \hbox{(odd)}\,,
\end{align*}
so $E(A)\cong A^{7|4}$ and  we can write
\begin{align*}
(r+\rho \ep_5) \wedge (s+\sigma \ep_5)  = & r \wedge s + (\sig r - \rho s) \wedge \ep_5+ \rho \sig \ep_5 \wedge \ep_5=
\\
& q + \lambda \wedge \ep_5 + d_{55} \ep_5 \wedge \ep_5\, ,
\end{align*}
with
\begin{align*}
&q:=d_{12}e_1 \wedge e_2+ d_{13}e_1 \wedge e_3 + d_{14}e_1 \wedge e_4+
d_{23}e_2 \wedge e_3 + d_{24} e_2 \wedge e_4 + d_{34}e_3 \wedge e_4,
\\
&\lambda:=\dd_{15}e_1 + \dd_{25}e_2+ \dd_{35}e_3+ \dd_{45}e_4 ,\\
&d_{55}:=\rho\sig=\alpha_{51}\alpha_{52}\,,
\end{align*}
and
$$d_{ij}:=\det \begin{pmatrix}a_{i1}&a_{i2}\\a_{j1}&a_{j2}\end{pmatrix},\qquad \dd_{i5}:=\det \begin{pmatrix}a_{i1}&a_{i2}\\\alpha_{51}&\alpha_{52}\end{pmatrix}\,.$$
Notice that although $\dd_{i5}$ is defined as a usual $2\times 2$ determinant,
it is indeed an odd element in $A$.

It is not difficult to prove that the coordinates $d_{ij}$, $d_{55}$,  $\dd_{i5}$,
determine uniquely a subspace $W_1(A) \in \Gr_1(A)$ (see also Section 4.8 in Ref. \cite{fl};  there the notation $\Gr^{\mathrm{ch}}$ is used instead of $\Gr_1$). Moreover,
if we change the basis we used to describe $W_1(A)$, that is, if we act on the right with $\rGL(2|0)=\rGL(2,\C)$  on $W_1(A)$,  these coordinates
vary by a common constant factor.
The natural transformation becomes
$$
\begin{CD}
 \Gr_1(A) @>p_A>>\bP(E)(A) \\
 W_1(A)=\rspan\left\{ a_1,a_2\right\}@>>>[a_1\wedge a_2]\,,\end{CD}
$$
where $\bP(E)(A)\cong \bP^{6|4}$. In coordinates, the map is
$$
\rspan\left\{ r+\rho \ep_5, \,s+\sigma \ep_5\right\rangle\rightarrow [d_{12}, d_{13},d_{14},d_{23},d_{24},d_{34},d_{55}; \delta_{15},\delta_{25},\delta_{35},
\delta_{45}]\,.$$
The projective embedding of $\Gr_1\subset \bP(E)$ defined above is called {\it the super Pl\"{u}cker embedding}.

\bigskip

We are going to characterize the image of the super Pl\"{u}cker embedding in terms of homogeneous polynomials. Then, we will have proven that $\Gr_1$ is a projective supervariety. We ask then when a generic, even vector $w$ in $E(A)$
\be w= d+\delta\wedge \ep_5+d_{55}e_5\wedge\ep_5\,,\label{genericv}\ee with
\begin{align*}
&d:=d_{12}e_1 \wedge e_2+ d_{13}e_1 \wedge e_3 + d_{14}e_1 \wedge e_4+
d_{23}e_2 \wedge e_3 + d_{24} e_2 \wedge e_4 + d_{34}e_3 \wedge e_4,
\\
&\delta:=\dd_{15}e_1 + \dd_{25}e_2+ \dd_{35}e_3+ \dd_{45}e_4\,.
\end{align*}
is {\it decomposable}, that is,  it can be written as a wedge product \be w= (r+\rho \ep_5) \wedge  (s+\sigma \ep_5)\,.\label{decomposable}\ee
One can prove \cite{cfl} that  this happens if and only if
 the following conditions are satisfied
$$
d\wedge d=0, \qquad d \wedge \delta=0, \qquad
\delta\wedge\delta=2d_{55}d\,.
$$
These equations are known as {\it the  super Pl\"ucker relations} for $\Gr_1$. They give  all the relations among the coordinates $d_{ij}, \delta_{i5}$ and
provide  a presentation of the coordinate ring of
$\Gr_1$ associated to this embedding. More explicitly, the super Pl\"{u}cker relations are
\begin{align}
&d_{12}d_{34}-d_{13}d_{24}+d_{14}d_{23}=0, &&
\hbox{ (classical Pl\"{u}cker relation)}
\nonumber  \\&d_{ij}\dd_{k5}-d_{ik}\dd_{j5}+d_{jk}\dd_{i5}=0,&&
1\leq i<j<k\leq 4,\nonumber \\
& \dd_{i5} \dd_{j5}=d_{55}d_{ij},&& 1\leq i<j\leq 4\,.\label{superplucker}
\end{align}
Let us denote as $I_{\Gr_1}$ the homogeneous ideal generated by the quadratic relations (\ref{superplucker}). Then, the Grassmannian coordinate ring $\C[\Gr_1]$
resulting from this embedding is given by
\be
\C[\Gr_1]=\C[d_{ij}, d_{55}; \dd_{i5}]/I_{{\Gr_1}}\,.
\label{embeddingring}\ee
The projective variety so defined  is called  the {\it super Klein quadric}, and it is isomorphic to $\Gr_1$ (see Refs. \cite{va,cfl}).

\medskip

We now turn to the problem of finding a projective embedding for
$\Gr_2$.
Let $W_2(A) \in \Gr_2(A)$ with
\be
W_2(A) = \rspan\left\{  \begin{pmatrix} a_{11} \\
a_{21} \\ a_{31} \\ a_{41} \\ \alpha_{51} \end{pmatrix}
\, , \,
\begin{pmatrix} a_{12} \\
a_{22} \\ a_{32} \\ a_{42} \\ \alpha_{52} \end{pmatrix}
\, , \,
\begin{pmatrix} \al_{15} \\
\al_{25} \\ \al_{35} \\ \al_{45} \\ a_{55} \end{pmatrix}
\right\} \subset A^{4|1}\,.\label{w2a}
\ee
As always, Latin letters are even elements in $A$ and Greek letters are odd elements in $A$.  One could try to mimic the procedure used for $\Gr_1$: there will appear $3\times 3$  determinants. But the change of basis is now given by the right action of the supergroup $\rGL(2|1)(A)$, and the ordinary determinants are not invariant under the action of the supergroup. The correct invariant objects in this case would be the Berezinians. This complicates considerably all the calculations.

 Nevertheless we can bypass this problem by using a natural duality between $\Gr_1$ and $\Gr_2$ that we are going to describe.

 Let us denote $T(A)=A^{4|1}$, and consider the   dual module $T(A)^*=\rhom_A(T(A), A)$, the set of (even and odd) $A$-linear maps  or forms
$A^{4|1}\rightarrow A$ (see Remark \ref{duals-rem}). We have that $T(A)\cong T(A)^*\cong A^{4|1}$, although the isomorphism is not natural. We will denote
$$\Gr_1^*(A):=\{\,\hbox{projective submodules }  M\subset T(A)^* \hbox{ of rank }2|0\,\}\,.$$ Clearly $\Gr_1^*\cong\Gr_1$.

Let $W_2(A)\in\Gr_2 (A)$, where $A$ is now an arbitrary superalgebra. We define the annihilator\footnote{We use here the word `annihilator' in the same sense than it is used for vector spaces.} of $W_2(A)$ as
$$W_2(A)^0=\{ \,u^* \in T(A)^*\; |\; u^*(v)=0\;\; \forall \; v\in W_2(A)\,\}\,.$$
 $W_2(A)^0$ is a submodule of $T(A)^*$ but in general, it is  not a projective submodule.

In spite of that, let us first see how the construction works for  local superalgebras; then $W_2(A)$ is a free module so it has a basis.
 Let
 \be \{e_1, e_2, e_3, e_4, \ep_5\}\label{basisT}\ee denote the canonical basis of $T(A)$ and let \be\{e_1^*, e_2^*, e_3^*, e_4^*, \ep_5^*\}\label{dualbasisT}\ee  denote the dual basis of $T(A)^*$  (\ref{dualbasis}).
 Suppose that $W_2(A)= \rspan\left\{ e_1, e_2, \ep_5\right\}$; then its annihilator is
$W_2(A)^0=\rspan\left\{ e^*_3, e^*_4\right\}$. For a general submodule $W_2(A)$, we can always choose a basis so that $W_2(A)$ and $W_2(A)^0$ have this form.  The change of basis accounts for the left action of the supergroup $\rSL(4|1)$, so for a general $W_2(A)$ we can write
  \begin{align}W_2(A)&=\langle \,g \cdot e_1, \,g \cdot e_2,\, g \cdot \ep_5 \,\rangle\subset T(A),\label{group1}\\W_2^0(A)&=\langle \,(g^t)^{-1} \cdot e_3^*,\,(g^t)^{-1} \cdot e^*_4\,\rangle\subset T(A)^*,\qquad g\in \rSL(4|1)(A)\,. \label{group2}\end{align}
We have then established a natural transformation between the functors $\Gr_2$ and $\Gr_1^*$  restricted to local algebras
$$\begin{CD}\Gr_2(A)@>q_A>> \Gr_1^*(A)\\
W_2(A)\subset T(A)@>>>W_2(A)^0\subset T(A)^*\,.\end{CD}$$ It has an inverse since there is a natural identification
$$T(A)^{**}\cong T(A)\,,$$ so the transformation is an isomorphism of functors.

  \begin{remark}
The fact that the entries of the
matrix $(g^t)^{-1}$ are expressed in terms of the entries of
$g$ has crucial consequences when
considering quantum deformations of these objects; in fact, this
will enable us to realize both $\Gr_1$ and $\Gr_2$ within the same
quantum matrix bialgebra.

\hfill$\square$
\end{remark}

 We still have to tackle the problem of extending  the natural transformation $q_A$ to generic superalgebras. Let $W_2(A) \in \Gr_2$, with $A$ generic, and let us consider it as an $A_0$-module.
Being finitely generated and projective,  it is {\it locally free}
 (see for example Ref. \cite{ei}, page 137 for the ordinary setting).
This means that there exists $f_i \in A_0$, $\, i=1,\dots, n$ such that the ideal $(f_1,\dots,f_n)=A_0$ and such that $W_2(A)[f_i^{-1}]$ is free as an $A[f_i^{-1}]$-module
(see
Theorem B.3.4 in Ref. \cite{ccf}).

 The geometric meaning of this fact is the following: since the set of prime ideals in $A$ is in bijective correspondence with the set of  prime ideals
in $A_0$ (see Chapter 10 in Ref. \cite{ccf}), the topological space  $\spec(A)=\spec(A_0)$
is covered by the open sets \be\spec(A_0)_{f_i}=\spec(A_0[f_i^{-1}])\,,\label{gluing}\ee and the localization of $W_2(A)$ on each $f_i$, $W_2(A)[f_i^{-1}]$ is a free $A[f_i^{-1}]$-module.
The annihilator  $W_2(A)[f_i^{-1}]^0$ can be constructed as in (\ref{group2}). We have then a collection of free $A[f_i^{-1}]$-modules $W_2(A)[f_i^{-1}]^0$
of rank $2|0$  which obviously agree on
$A[f_j^{-1}][f_i^{-1}]$, and such that the gluing property (\ref{gluing}) is satisfied.
They then glue to a projective $A$-module that we will denote as  $W(A)^\perp\in \Gr^*_1(A)$. The construction is manifestly functorial. For more details on this construction, see Section 3.16 in Ref. \cite{lam}  and also  Sections 10.2.3 and 10.3.2 in Ref. \cite{ccf}.

We have then established a natural transformation (that we denote also as $q_A$)
$$\begin{CD}\Gr_2(A)@>q_A>> \Gr_1^*(A)\\
W_2(A)\subset T(A)@>>>W_2(A)^\perp\subset T(A)^*\,,\end{CD}$$ which is an isomorphism of functors on arbitrary superalgebras.

\medskip

We want now to obtain the super Pl\"{u}cker relations as in (\ref{superplucker}). We consider   $W_2(A)^\perp$  as a submodule of $T(A)^*$. In terms of the dual basis (\ref{dualbasisT}),
we will denote a vector  in $T(A)^*$ simply as
$$ r^*+\rho^* \ep^*_5,\qquad \hbox{with} \qquad r^*=a^*_{1}e^*_1+a^*_{2}e^*_2+a^*_{3}e^*_3+a^*_{4}e^*_4\,,$$ and a vector in $E(A)=\wedge^2\C^{4|1}(A)$ as
\be w^*= d^*+\delta^*\wedge \ep_5^*+d_{55}^*e_5^*\wedge\ep_5^*\,,\label{dualgenericv}\ee with
\begin{align*}
&d^*:=d^*_{12}e^*_1 \wedge e^*_2+ d^*_{13}e^*_1 \wedge e^*_3 + d^*_{14}e^*_1 \wedge e^*_4+
d^*_{23}e^*_2 \wedge e^*_3 + d^*_{24} e^*_2 \wedge e^*_4 + d^*_{34}e^*_3 \wedge e^*_4,
\\
&\delta^*:=\dd^*_{15}e^*_1 + \dd^*_{25}e^*_2+ \dd^*_{35}e^*_3+ \dd^*_{45}e^*_4\,.
\end{align*}
the two-form $w^*$ is decomposable
\be w^* =(r^*+\rho^* \ep^*_5)\wedge(s^*+\sigma^*\ep^*_5) \label{dualdecomposable}\ee if the   super Pl\"{u}cker relations
\begin{align}
&d^*_{12}d^*_{34}-d^*_{13}d^*_{24}+d^*_{14}d^*_{23}=0, &&
\hbox{ (classical Pl\"{u}cker relations)}
\nonumber  \\&d^*_{ij}\dd^*_{k5}-d^*_{ik}\dd^*_{j5}+d^*_{jk}\dd^*_{i5}=0,&&
1\leq i<j<k\leq 4,\nonumber \\
& \dd^*_{i5} \dd^*_{j5}=d^*_{55}d^*_{ij},&& 1\leq i<j\leq 4
\label{dualsuperplucker}
\end{align}
are satisfied.
\medskip

We have proven the following theorem:

\begin{theorem}
\label{projembgrass}
The product of super Grassmannians $\Gr_1$ and $\Gr_2$ is
embedded into the product of super projective spaces:
$$
\Gr_1 \times \Gr_2 \subset \bP(\wedge^2 T) \times \bP(\wedge^2 T^*)
$$
With respect to such projective embedding, the coordinate
ring of $\Gr_1 \times \Gr_2$ is given by:
$$
\C[d_{ij}, d_{kl}^*,d_{55}, d^*_{55}, \delta_{m5}, \delta^*_{n5}] \, \big/ \,
(I_{\Gr_1}+I_{\Gr_2})\,,
$$
where $I_{\Gr_1}$ is the ideal of the super Pl\"ucker relations (\ref{superplucker}),
while $I_{\Gr_2}$ is the ideal of the super Pl\"ucker
relations (\ref{dualsuperplucker}).

\hfill$\square$
\end{theorem}

It is possible to see the coordinate rings of $\Gr_1$ and $\Gr_2$ as subrings of the coordinate ring of $\rSL(4|1)$. This is a consequence of (\ref{group1}) and (\ref{group2}), and it is a crucial point in the quantization that we will propose in Section \ref{qsuperflag-sec}. Let $A$ be a local superalgebra and $g\in \rSL(4|1)(A)$
$$g=\begin{pmatrix}
g_{11}&g_{12}&g_{13}&g_{14}&\gamma_{15}\\
g_{21}&g_{22}&g_{23}&g_{24}&\gamma_{25}\\
g_{31}&g_{32}&g_{33}&g_{34}&\gamma_{35}\\
g_{41}&g_{42}&g_{43}&g_{44}&\gamma_{45}\\
\gamma_{51}&\gamma_{52}&\gamma_{53}&\gamma_{54}&g_{55}
\end{pmatrix}\,.$$
Let now
$$e_1=\begin{pmatrix}1\\0\\0\\0\\0\end{pmatrix},\qquad e_2=\begin{pmatrix}0\\1\\0\\0\\0\end{pmatrix}$$ be two vectors of the standard basis.
Then the action of $g$ on $e_1$ and $e_2$ selects the first two columns of $g$
$$ge_1=\begin{pmatrix}
g_{11}\\
g_{21}\\
g_{31}\\
g_{41}\\
\gamma_{51}
\end{pmatrix},\qquad ge_2=\begin{pmatrix}
g_{12}\\
g_{22}\\
g_{32}\\
g_{42}\\
\gamma_{52}
\end{pmatrix}\,,$$ and these  are the two independent vectors generating the subspace $W_1(A)$. The coordinate ring of
$\rSL(4|1)$ is
$$\C[\rSL(4|1)]=\C[g_{ij}, \gamma_{i5}, \gamma_{5i}]/(\Ber\, g-1)\,,$$ where we are now interpreting $g_{ij}, \gamma_{i5}$ and $\gamma_{5i}$ as (even and odd) indeterminates and not as elements of the superring $A$. One can show \cite{cfl, cfl2, fl} that $\C[\Gr_1]$ is the subring of $\C[\rSL(4|1)]$ generated by the $2\times 2$ determinants
$$d^{12}_{ij}:=g_{i1}g_{j2}-g_{i2}g_{j1},\qquad \delta^{12}_{i5}:=g_{i1}\gamma_{52}-g_{i2}\gamma_{51}, \quad d^{12}_{55}=\gamma_{51}\gamma_{52}\,,$$ with $i,j=1,\dots, 4$. We will suppress the superindex indicating columns 1 and 2, which coincides with the notation above. Nevertheless, we have to remember that when seeing $\C[\Gr_1]$ as a subalgebra of $\C[\rSL(4|1)]$, these expressions refer to the determinants of the first two columns of the generators $g_{ij}$ of $\C[\rSL(4|1)]$.

Because of the duality between $\Gr^*_1$ and $\Gr_2$ discussed above there are the corresponding expressions for $\Gr_2$, but now one has to consider the element of the group $(g^t)^{-1}$ (see (\ref{group2})). We  are going to give these expressions explicitly, but we need first some notation.

\begin{notation}\label{notationindices} Let us denote columns with  upper indices and rows with  lower indices. Let $I=(i_1,\dots, i_p)$ and $J=(j_1,\dots, j_p)$ be multiindices. Then
$d^J_I$
stands for usual determinant  obtained by taking the columns
$J$ and the rows $I$, while $\delta^J_I$ stands for usual determinant
with an odd column (in our case, the only possibility is the $5^{\mathrm{th}}$ one).
$B$ denotes the total Berezinian of $g$, while $b^J_I$ stands for Berezinian
obtained by taking columns $J$ and rows $I$.
We also write $b^{125}_{1 \dots \hat {\imath} \dots \hat {\jmath} \dots 5}$ for the Berezinian
obtained by taking the columns $1,2,5$ and the rows obtained by removing
$i$ and $j$ from $1,2,3,4,5$.

\hfill$\square$

\end{notation}

We have then the following proposition:

\begin{proposition}
Let
$$g = \begin{pmatrix} p_{4\times 4} & q_{4\times 1} \\ r_{1\times 4} & s_{1\times 1} \end{pmatrix}\in \rSL(4|1)(A), \qquad \hbox{or simply}\qquad  g=\begin{pmatrix} p & q \\ r & s \end{pmatrix}
$$  as in Example \ref{algebraicsupergroup}. The rest of the notation is as above. Then
\begin{align*}
d_{ij}^* & := d^{34}_{ij}((g^t)^{-1})=(-1)^{p(s)}\,
b^{125}_{1 \dots \hat {\imath} \dots \hat {\jmath} \dots 5}, \\ \\
\delta_{i5}^*&:=d^{34}_{i5}((g^t)^{-1})=(-1)^i \,\frac{B}{ \det( p)^2}\,
\left[ b^{1245}_{1 \dots \hat {\imath} \dots 45} \,\delta^{1235}_{1234} -
b^{1235}_{1 \dots \hat {\imath} \dots 45} \,\delta^{1245}_{1234} \right], \\ \\
d_{55}^* &:= ((g^t)^{-1})_{51} ((g^t)^{-1})_{52} =
-\frac{B^2}{ \det( p)^4} \,\delta^{2345}_{1234} \, \delta^{1345}_{1234}\,.
\end{align*}
where $s$ is the permutation
$$(1\,2\,3\,4\,5) \rightarrow (i\,j\,u\,v\,w)\qquad \hbox{with} \qquad
\{u,v,w\}=\{1, \dots  \hat{\imath}, \dots, \hat{\jmath}, \dots 5\}\,.$$
\end{proposition}

\begin{proof}
This can be checked by a  long computation. We sketch the  proof for the  particular case $d^*_{12}$. The other cases go along similar lines.

Let us define ${\bf p}:=p-q s^{-1} r$ and ${\bf s}:=s-r p^{-1} q$, we then have, explicitly
\\
$$(g^{-1})^t = \begin{pmatrix} ({\bf p}^{-1})^t & -(s^{-1}r{\bf p})^t \\ -(p^{-1}q{\bf s})^t& {\bf s}^{-1} \end{pmatrix}\,.
$$
The key identity is given by
$$
d_{12}^{34}(({\bf p}^t)^{-1})=\frac{1}{\det {\bf p} }d^{12}_{34}({\bf p})\,.
$$
For $g \in \rSL(4|1)(A)$ we have ${\bf s}= \det p$ and $s=\det {\bf p}$, thus
$$
d_{12}^{34}(({\bf p}^t)^{-1})=\frac{1}{s} \, d^{12}_{34}\,(p-q s^{-1} r)= b^{125}_{345}
$$
as we wanted to show.

\end{proof}

\subsection {The Pl\"{u}cker embedding of the super flag \\ $\Fl(2|0,2|1;4|1)$}\label{superflag-sec}
In the same way than for the projective superspace and  the super Grassmannians, we define the functor of points of the flag supervariety for a superalgebra $A$ as
\begin{align*}\F(A)=\{\,&\hbox{finitely generated, projective submodules }  M\subset N\subset  A^{4|1}\\& \hbox{ where } M \hbox{ is of rank } 2|0 \hbox{ and } N \hbox{ is of rank } 2|1\,\}\,. \end{align*}
This functor, when expressed in terms of superschemes, is also representable \cite{fl}. The reduced scheme is again the Grassmannian $\Gr(2,4)$. One can see immediately that there is a left action action of the supergroup $\rSL(4|1)$. As an homogeneous space, the functor is given as
$$\F=\rSL(4|1)/ P_u,\qquad P_u=P_1\cap P_2 \,,$$ where $P_1$ and $P_2$ are the parabolic subgroups defined in (\ref{parabolics}). We have
\be P_u(A)=\left\{\begin{pmatrix}
g_{11} &g_{12}&g_{13}&g_{14}&\gamma_{15}\\
g_{21} &g_{22}&g_{23}&g_{24}&\gamma_{25}\\
0 &0&g_{33}&g_{34}&0\\
0 &0&g_{43}&g_{44}&0\\
0 &0&\gamma_{53}&\gamma_{54}&g_{55}\end{pmatrix}\right\}\,.\label{interparabolic}\ee

As for the super Grassmannians, if $A$ is local, one can express the functor of points of the flag supervariety in terms of the action of $\rSL(4|1)$ on the canonical basis, namely
$$\F(A)=\left\{\left(\rspan\left\{\, g\cdot e_1, \,g\cdot e_2\,\right\} \,,\,
\rspan\left\{\, g\cdot e_1, \,g\cdot e_2\,\, g\cdot \ep_5\right\}\right)\,
,\; g\in \rSL(4|1)(A)\,\right\}$$
(see Sections  4.11 and 4.12 in Ref. \cite{fl}).

We want to determine the conditions on an element of $\Gr_1(A) \times \Gr_2(A)$ to actually belong to the superflag, that is, $(W_1(A), W_2(A))\in \Gr_1(A) \times \Gr_2(A)$ be such that $W_1(A)\subset W_2(A)$ (we remind the reader that $A$ is still a local algebra).
This will determine extra relations among the generators of the the coordinate ring of $ \Gr_1 \times \Gr_2$  (see Theorem \ref{projembgrass})  needed to describe the superflag with its projective embedding.

We denote
\begin{align*}W_1(A)&=\rspan\left\{ a,b\right\}=\rspan\left\{ \begin{pmatrix}a_1\\a_2\\a_3\\a_4\\\alpha_5\end{pmatrix}, \begin{pmatrix}b_1\\b_2\\b_3\\b_4\\\beta_5\end{pmatrix}\right\}\\\\
W_2^0(A)&=\rspan\left\{ c^*,d^*\right\}=\rspan\left\{ \begin{pmatrix}r_1^*\\r_2^*\\r_3^*\\r_4^*\\\rho_5^*\end{pmatrix}, \begin{pmatrix}s_1^*\\s_2^*\\s_3^*\\s_4^*\\\sigma_5^*\end{pmatrix}\right\}\,.
\end{align*} The necessary and sufficient conditions for $W_1(A)\subset W_2(A)$ are
\be r^*(a)=0,\quad r^*(b)=0,\quad s^*(a)=0,\quad s^*(b)=0\,.\label{incidence1}\ee

There is a way of expressing (\ref{incidence1}) in terms of $d=a\wedge b$ and $ d^*=r^* \wedge s^* $. First, we unify the notation by defining $e_5:=\ep_5$ and $e_5^*:=\ep_5^*$.
Then we define the contraction
$$\begin{CD} \left(T(A)^*\otimes T(A)^*\right)\otimes \left(T(A)\otimes T(A)\right)@>\Phi>> T(A)^*\otimes T(A)\\
\left( p^*_{ij} e^*_i\otimes  e^*_j\right)\otimes\left(q_{ij} e_i\otimes \ e_j\right)@>>> \sum_{j=1}^5 p^*_{ij}q_{jk} e^*_i\otimes e_k\,.\end{CD}$$
Assuming that $d^*$ and $d$ are decomposable, (\ref{incidence1}) is equivalent to
\be\Phi( d^*\otimes d)=0\,.\label{incidence2}\ee
 In components,
\be  \sum_{j=1}^5d^*_{ij}d_{jk}=0, \qquad \forall\; i,k=1,\dots, 5\,. \label{incidence3}\ee where $d_{i5}:=\delta_{i5}$ and $d_{i5}^*:=\delta_{i5}^*$
Notice that this expression is antisymmetric in $(i,j)$ and $(j,k)$, but the sum over $j$ is unrestricted.
This gives a total of 25 independent conditions that we write explicitly in Appendix \ref{incidence-ap}.
They are the {\it incidence relations}. Conditions (\ref{superplucker}), (\ref{dualsuperplucker}) and (\ref{incidence3})  define completely the flag supervariety $\F$.
We then can state the analog of Theorem \ref{projembgrass}
for the superflag.

\begin{theorem} \label{embflag}
There is an embedding of the superflag
$$
\F \subset \Gr_1 \times \Gr_2 \subset \bP(E)
\times \bP(E^*)\,,
$$
with $E=\wedge^2(T) \cong \C^{7|4}$, $E^*=\wedge^2(T^*) \cong \C^{7|4}$.
With respect to such embedding, the bigraded coordinate ring of $F$ is
given by
$$
\C[F]:=\C[d_{ij},\dd_{i5}, d_{55}, d_{ij}^*, \dd_{i5}^*, d_{55}^*] \, \big/ \,
(I_{\Gr_1}+I_{\Gr_2}+I_{\mathrm{inc}})
$$
where $I_{\mathrm{inc}}$, is the ideal generated by the 25 incidence relations in Appendix \ref{incidence-ap}.

\hfill$\square$
\end{theorem}

\begin{remark} \label{remarksubring} As for $\C[\Gr_1]$ and $\C[\Gr_2]$, also the coordinate ring of the superflag $\C[\F]$ relative to this embedding is the subring in $\C[\rSL(4|1)]$ generated by the determinants $d_{ij},\dd_{i5}, d_{55}, d_{ij}^*, \dd_{i5}^*$ and $d_{55}^*$ inside $\C[\rSL(4|1)]$. While dealing at the same time with coordinates $d_{ij}$ and $d^*_{ij}$ is already cumbersome at the classical level, at the quantum level the problem becomes intractable. This is  why we will  resort to a different argument, based on super line bundles, to define the ring of the quantum superflag in Sections \ref{qsection-sec} and \ref{qsuperflag-sec}. Nevertheless, the whole construction is based on the same property for $\C[\Gr_1]$ and $\C[\Gr_2]$.

\hfill$\square$

\end{remark}

In the following section  we are going to see how the product of two projective superspaces can be embedded into one projective superspace of higher dimension.

\subsection{The super Segre map} \label{segre-sec}
In this section we extend the Segre map to the super setting and then apply it to the case of interest for the embedding of the superflag $\F$ into a unique projective space.

In the ordinary case, given two projective spaces
$\bP^n(\C)$, also denoted as $\bP(\C^{n+1})$
, the Segre map is given by
$$
\begin{CD}
  \bP(\C^{n+1})\times \bP(\C^{d+1})  @>\psi>>\bP(\C^{n+1}\otimes\C^{d+1}) \\
(U, V)@>>>U\otimes V\,.
\end{CD}
$$
 In terms of the homogeneous coordinates for each space, one has
$$
\begin{CD}
  \bP^n(\mathbb{C})\times \bP^d(\mathbb{C})  @>\psi>>\bP^N(\mathbb{C}) \\
([x_0,\dots,x_n],[y_0,\dots,y_d])@>>>[x_0y_0, x_0y_1,
\dots, x_iy_j, \dots, x_ny_d]
\end{CD}
$$
with $\,i=0,..,n\,$, $\,j=0,..,d\,$ and $\,N=(n+1)(d+1)-1$.

As usual, $[x_0,\dots,x_n]$ stands for the equivalence class in $\bP^n(\C)$
given by $(x_0,\dots,x_n)\sim\lambda(x_0,\dots,x_n)$, for all $\lambda \in \C^\times$ (see Example \ref{fop projective superspace}).

If we label as $z_{ij}$ the homogeneous coordinates of $\bP^N(\C)$,
then it is not hard to prove that the image of the Segre map is an
algebraic variety given by the zero locus of the $2\times 2$
minors of the matrix (see for example Chapter 1 in  Ref. \cite{ha})
$$
\left(
\begin{matrix}
z_{00}&z_{02}&\cdots&z_{0d}\\
z_{20}& &&\\
\vdots&&\ddots&\\
z_{n0}&\cdots&&z_{nd}
\end{matrix}\,.
\right)
$$
Moreover, the map $\psi$ is an embedding.

\medskip

In order to generalize this construction to the super setting we have to define a natural transformation among the functors
$
\Psi: \bP^{n|r}\times \bP^{d|s} \rightarrow\bP^{N|M}
$,
where $N=(n+1)(d+1)+rs-1\,$ and   $\,M=(n+1)s+(d+1)r$. For an arbitrary superalgebra $A$, this is given by
\be\begin{CD}\bP^{n|r}(A)\times \bP^{d|s}(A) @>\Psi_A>>\bP^{N|M}(A)\\
(U(A), V(A))@>>> U(A)\otimes_A V(A)\,. \end{CD}\label{supersegre}\ee Since the natural transformation is defined in general, we can now restrict to local algebras. The submodules are free and in homogeneous coordinates we have

 \be
\begin{CD}
  \bP^{n|r}(A)\times \bP^{d|s}(A)  @>\Psi_A>>\bP^{N|M}(A) \\
([x_0,\dots,x_n\,|\,\alpha_1,\dots,\alpha_r],[y_0,\dots,y_d\,|\,\beta_1,\dots,\beta_s])@>>>
[x_iy_j,\alpha_k\beta_l\,|\,x_i\beta_l,y_j\alpha_k]\,.
\end{CD}\label{supersegrerita}
\ee
with
$$
i=0, \dots , n,\qquad j=0,\dots, d,\qquad
k=1,\dots, r,\qquad l=1,\dots, s\,.
$$

The super Segre map is an embedding and one can give the polynomial equations defining its image. To do so, we use  the {\it even rules principle} \cite{dm}, a technique essentially similar to the usual `Grassmann envelope' well known to people working with superalgebras, that helps to keep track of insidious signs. In order not to disrupt the discourse at this point we have preferred to put that proof in the Appendix \ref{embedding-ap}, where we also describe the technique.

\medskip

We have shown in Section \ref{superflag} how the superflag $\F$ embeds into the product of two
projective superspaces through the Pl\"{u}cker embedding. We just need now to consider the composition of the natural transformation of Theorem \ref{embflag} with the super Segre map (\ref{supersegre}).  Restricting to local superalgebras we have
\be \label{embedding-F}
\begin{CD}
\F(A)@>>> \bP(E)(A) \times \bP(E^*)(A) \\
W_1(A) \subset W_2(A) @>>>([d_{ij}, d_{55}\,|\,\dd_{i5}],[ d_{ij}^*,  d_{55}^*\,|\,\dd_{i5}^*])
\end{CD}\ee
where $W_1(A)=\rspan\{ r+\rho\ep_5, s+\sigma\ep_5\}$ and $W_2(A)$ is given in terms of $W_2(A)^0=\rspan\{ r^*+\rho^*\ep^*_5, s^*+\sigma^*\ep^*_5\}$. The relation with the coordinates $d_{ij},\dots d_{ij}^*,\dots $ is then  expressed   in equations (\ref{genericv}, \ref{decomposable}, \ref{dualgenericv}, \ref{dualdecomposable}).

Composing with the super Segre map we will embed $F$
into $\bP^{M|N}$,
where $M|N=64|56$; explicitly, we get:
\begin{equation} \label{embedding2}
\begin{CD}
 \bP(E)(A) \times \bP(E^*)(A) @>\Psi>> \bP^{M|N}(A) \\
([z_{ij,} z_{55}\,|\,\zeta_{i5}],[ z_{ij}^*,  z_{55}^*\,|\,\zeta_{i5}^*])
@>>>[ z_{ij}z_{kl}^*,z_{55}z_{55}^*,z_{ij}z_{55}^*,z_{55}z_{kl}^*,
\zeta_{i5}\zeta_{k5}^*\,|\,\\@. z_{ij}\zeta_{k5}^*,z_{55}\zeta_{k5}^*,\zeta_{i5}z_{kl}^*,\zeta_{i5}z_{55}^*]\,.  \\
\end{CD}
\end{equation}

Let us denote
$I, K=(1,2), (1,3),  (1,4), (2,3), (2,3), (3,4)$: Then, according to the notation in Appendix \ref{embedding-ap}, we can organize the image of the super Segre map $\C^{M+1|N}$ in matrix form:
\be\left(\begin{array}{cc}
z_Iz^*_K&z_{55}z^*_K\\
z_Iz^*_{55}&z_{55}z^*_{55}\\
\hline
\zeta_{i5}z^*_K&\zeta_{i5}z^*_{55}\end{array}\vline
\begin{array}{c}
z_I\zeta^*_k\\
z_{55}\zeta^*_k\\
\hline
\zeta_{i5}\zeta^*_{k5}\label{imagesegre}
\end{array}\right)\,.
\ee
This image is a projective algebraic variety in the  generators

$$\left(\begin{array}{cc}
Z_{IK}&Z_{5K}\\
Z_{I5}&Z_{55}\\
\hline
\Gamma_{iK}&\Gamma_{i5}\end{array}\vline
\begin{array}{c}
\Lambda_{Ik}\\
\Lambda_{5k}\\
\hline
T_{ik}
\end{array}\right)\,,
$$
satisfying the homogeneous polynomial relations (\ref{polynomials2}).

\section{Line bundles and projective embeddings.}\label{linebundles-sec}

In this section we want to introduce basic concepts in parabolic geometry and construct the very ample line bundle describing the projective embedding of the Grassmanian. Along the way, we will also discuss its interpretation as the bundle of conformal densities. We will then
extend our construction to the super Grassmanians $\Gr_1$ and $\Gr_2$ as well as the super flag $\F$.
\subsection{Parabolic (super)geometries
}
\label{homogeneous}

Let $G$ be a semisimple Lie group, $\fg=\mathrm{Lie}(G)$ and we consider its root decomposition with respect to a Cartan subalgebra $\fh$. If $\Delta^\pm$ are the subsets of positive and negative roots respectively, with $\Delta=\Delta^++\Delta^-$, we denote as usual
$$\fn_+=\sum_{\alpha\in \Delta^+}\fg_\alpha,\qquad \fn_-=\sum_{\alpha\in \Delta^-}\fg_\alpha\,,$$ so we have the decomposition
 $\fg=\mathfrak{n}_-\oplus \mathfrak{h}\oplus \mathfrak{n}_+$, The {\it Borel subalgebra} for this system of roots is
 $\mathfrak{b}_\pm=\mathfrak{h}\oplus \mathfrak{n}_\pm$, where one can choose indifferently $\fb_+$ or $\fb_-$. All the Borel subalgebras are conjugated.
 A {\it parabolic subalgebra} of $\fg$ is a subalgebra $\fp$ that contains the Borel subalgebra but it is not the full $\fg$.

 To every parabolic subalgebra $\fp$ there is  associated a $|k|$-grading of $\fg$
$$\fg=\fg_{-k}\oplus \cdots  \fg_{-1}\oplus\fg_0\oplus \fg_{1}\cdots \oplus\fg_{k},\qquad k\in \N\,,$$
with $[\fg_i,\fg_j]\subset \fg_{i+j}$, such that $\fp =\fg_0\oplus \fg_{1}\cdots \oplus\fg_{k}$. Also, one defines $\fp_+:= \fg_{1}\oplus\cdots \oplus\fg_{k}$ and $\fp_-:= \fg_{-k}\oplus\cdots \oplus\fg_{-1}$ .

 The {\it Levi subgroup} of $P$ is the group whose adjoint action preserves the grading. Its Lie algebra is $\fg_0$,  and we will denote it as $G_0$.  It is in fact the reductive component in the Levi decomposition, $P=G_0\ltimes P_+$   where $P_+$ is the unipotent radical with $\Lie(P_+)=\fp_+$.

 The following is the relevant example.

\begin{example} {\it The complexified conformal and Poincar\'{e} groups.} We  consider the conformal group of the Minkowski spacetime, the group $\rSO(4,2)$. Its spin group (the double covering) is $\rSU(2,2)$, with complexification $\rSL_4(\C)$. The Lie algebras $\fso_6(\C)\cong\fsl_4(\C)$ are isomorphic, but we will use the four dimensional notation. This means that the conformal algebra and the conformal group act on a four dimensional space called the  {\it twistor space} whose relation with the four dimensional spacetime we will see in a moment. For the standard choice of roots (diagonal), the Borel subalgebra consists of the lower triangular matrices\footnote{Or, alternatively, the upper triangular matrices.}
 $$\begin{pmatrix}*&0&0&0\\
 *&*&0&0\\
 *&*&*&0\\
 *&*&*&*\end{pmatrix}\,.$$
 The complexified {\it Poincar\'{e} subalgebra plus dilations} is a parabolic subalgebra consisting of lower, $2\!\times\! 2\,$-block triangular matrices
$$\begin{pmatrix}l&0\\m&r \end{pmatrix}\,,$$
 where the diagonal blocs $l$ and $r$ form the Lorentz subalgebra plus dilations $\fsl_2(\C)\oplus \fsl_2(\C)\oplus \C\cong \fso(4,\C)\oplus \C$, and the block $m$ represents the translations. The {\it Poincar\'{e} group times dilations}, in this context, is the group $P=(\rSL_2(\C)\times \rSL_2(\C)\times \C)\ltimes \rM_2(\C)$, where $ \rM_2(\C)$ is the space of $2\times 2$-dimensional matrices. We will denote a generic element of this group as
\be g=\begin{pmatrix}L&0\\NL&R \end{pmatrix}\,,\label{poincaregroup}\ee
(notice that $L$ and $R$ are invertible matrices).

The conformal space is the Grassmannian $\Gr(2,4)\cong\rSL_4(\C)/P_u$, where $P_u$ is the upper parabolic subgroup of elements
$$g=\begin{pmatrix}L&Q\\0&R \end{pmatrix}\,,$$
 which is conjugated to (\ref{poincaregroup}).
 The grading of the Lie algebra is
\begin{align*} &\fg_0=\left\{\begin{pmatrix}l&0\\0&r\end{pmatrix}\right\},\qquad
&\fg_{-1}=\left\{\begin{pmatrix}0&0\\m&0\end{pmatrix}\right\}, \qquad
\fg_{1}=\left\{\begin{pmatrix}0&q\\0&0\end{pmatrix}\right\}\,.
\end{align*}
The Levi subgroup $G_0$ is the Lorentz group times dilations. 
%

If one represents an element in the Grassmannian $\Gr(2,4)$ in terms of a $4\times 2$-matrix whose columns are the basis vectors of the 2-plane,  the big cell is characterized in terms of matrices with the minor $d_{12}\neq0$.  By a change of basis one can always bring such matrix to a standard form \be\begin{CD}\begin{pmatrix}a^1&b^1\\
a^2&b^2\\
a^3&b^3\\
a^4&b^4\end{pmatrix}@>>> \begin{pmatrix}\rid \\
A\end{pmatrix}\end{CD}\,,\label{minkowski}\ee and $A$ is an arbitrary $2\times 2$-matrix. The  action of the group on the big cell  is simply
\be A\longrightarrow RAL^{-1} + N,\qquad A\in M_2(\C)\,.\label{actionpoincare}\ee
A linear change of coordinates in terms of the Pauli matrices allows us to see the relation with the  standard coordinates in Minkowski space:
$$
\sigma_0=\begin{pmatrix} 1 & 0 \\  0  &  1 \end{pmatrix}, \quad
\sigma_1=\begin{pmatrix} 0 & 1 \\  1  &  0\end{pmatrix}, \quad
\sigma_2=\begin{pmatrix} 0 & -i \\ i   &  0 \end{pmatrix}, \quad
\sigma_3=\begin{pmatrix} 1 & 0 \\  0  &  -1 \end{pmatrix}\,,$$
with $A=\sum_{\mu=0}^4x^\mu\sigma_\mu$ and
$$\det A=(x^0)^2-(x^1)^2-(x^2)^2-(x^3)^2\,.$$
One can also check that the action of the Poincar\'{e} group times dilations in these coordinates is the expected one.

\hfill$\square$

\end{example}

Given a flag manifold on $\C^n$, the subgroup of $\rSL_n(\C)$ that stabilizes one point is an upper parabolic subgroup (that is, a subgroup of $\rSL_n(\C)$ whose Lie algebra is an upper parabolic subalgebra), so we have that a flag manifold can be written always a a quotient $\rSL_n(\C)/P_u$ (here $P_u$ stands for a generic upper parabolic subgroup). One can construct a generalized Pl\"{u}cker embedding for all the flags. Moreover, this construction generalizes for semisimple groups $G$ other than $\rSL_n(\C)$. The spaces $G/P_u$ are called {\it generalized flag varieties}. One has the following theorem:

\begin{theorem}\label{projectivetheorem}
Let $\fg$ be a complex semisimple Lie algebra, $\fp$ a parabolic subalgebra, $G$ a connected Lie group with Lie$(G)=\fg$ and $P$ a parabolic subgroup of $G$ with  Lie$(P)=\fp$. Then, the generalized flag manifold $G/P$ is a compact K\"{a}hler manifold and a projective algebraic variety. 

\hfill$\square$

\end{theorem}

We do not give here the proof of this well known result. The reader can consult, for example, Ref. \cite{cs}, page 306. A stronger result is in fact true: the quotient $G/P$, $P$ being a closed subgroup of $G$, is parabolic if and only if $G/P$ is projective (see Ref. \cite{bo}, Chapter IV).

\medskip

We turn now to the super case.  As mentioned in Section \ref{superflag}, Theorem \ref{projectivetheorem} does not have an extension to the super setting. Nevertheless,  the three superflags of interest for us, $\Gr_1$, $\Gr_2$ and $\F$ are projective, as shown explicitly with the super Pl\"{u}cker and super Segre embeddings.

As in the non super case, for each parabolic subalgebra of a superalgebra $\fg$ there is associated a $|k|$-grading of $\fg$ \cite{io}. We first analyze the  gradings corresponding to the three parabolic subalgebras.

\paragraph{\bf Gradings of $\fsl(4|1)$.}
Let us write in block form the super Lie algebra
$$\fg=\fsl(4|1)=\left \{\begin{pmatrix}
l&q&\nu\\
p&r&\alpha\\
\mu&\beta&s\\
\end{pmatrix}\;\; \Big|\;\; \tr\, l+\tr\, r=s\right\}\,,$$ where, as always, Latin letters denote blocks with even entries and Greek letters denote blocks with odd ones.

\begin{itemize}

\item For the parabolic subalgebra of the parabolic subgroup  $P_1$  in (\ref{parabolics}), associated to  the supergrassmannian $\Gr_1=
\rSL(4|1)/ P_1$:
$$
\fp_1 =\left\{
\begin{pmatrix}
l & q & \nu \\
0 & r & \alpha \\
0 & \beta & s
\end{pmatrix}  \right\}\,,
$$
we have the $|1|$-grading \begin{align*}\fg_{-1}=& \left
\{\begin{pmatrix}
0&0&0\\
p&0&0\\
\mu&0&0\\
\end{pmatrix}\right\},&
\fg_{0}=\left \{\begin{pmatrix}
l&0&0\\
0&r&\alpha\\
0&\beta&s\\
\end{pmatrix}\right\},&\quad
\fg_{+1}=&\left \{\begin{pmatrix}
0&q&\nu\\
0&0&0\\
0&0&0\\
\end{pmatrix}\right\}\,,&\quad
\end{align*}
with $\fp_1=\fg_0\oplus\fg_{+1}$.
\item For the parabolic subalgebra of the parabolic subgroup  $P_2$  in (\ref{parabolics}), associated to  the supergrassmannian $\Gr_2=
\rSL(4|1)/ P_2$:
$$
\fp_2 =\left\{
\begin{pmatrix}
l & q & \nu \\
0 & r & 0 \\
\mu & \beta & s
\end{pmatrix}  \right\}\,,
$$
we have the $|1|$-grading \begin{align*}\fg_{-1}=& \left
\{\begin{pmatrix}
0&0&0\\
p&0&\alpha\\
0&0&0\\
\end{pmatrix}\right\},&
\fg_{0}=\left \{\begin{pmatrix}
l&0&\nu\\
0&r&0\\
\mu&0&s\\
\end{pmatrix}\right\},&\quad
\fg_{+1}=&\left \{\begin{pmatrix}
0&q&0\\
0&0&0\\
0&\beta&0\\
\end{pmatrix}\right\}\,.&\quad
\end{align*}
with $\fp_2=\fg_0\oplus\fg_{+1}$.
\item
For the superalgebra the supergroup $P_u=P_1\cap P_2$ (\ref{interparabolic}, \ref{parabolics}) associated to the superflag $\F$:
$$\fp_u=\left \{\begin{pmatrix}
l&q&\nu\\
0&r&0\\
0&\beta&s\\
\end{pmatrix}\right\}$$
we have the $|2|$-grading
\begin{align*}\fg_{-2}=& \left
\{\begin{pmatrix}
0&0&0\\
p&0&0\\
0&0&0\\
\end{pmatrix}\right\},&\quad
\fg_{-1}=&\left \{\begin{pmatrix}
0&0&0\\
0&0&\alpha\\
\mu&0&0\\
\end{pmatrix}\right\},&
\fg_{0}=\left \{\begin{pmatrix}
l&0&0\\
0&r&0\\
0&0&s\\
\end{pmatrix}\right\},\\
\fg_{+2}=&\left \{\begin{pmatrix}
0&q&0\\
0&0&0\\
0&0&0\\
\end{pmatrix}\right\},&\quad
\fg_{+1}=&\left \{\begin{pmatrix}
0&0&\nu\\
0&0&0\\
0&\beta&0\\
\end{pmatrix}\right\}\,,&\quad
\end{align*}
with $\fp_u=\fg_0\oplus\fg_{+1}\oplus \fg_{+2}$.

\end{itemize}
 Notice that in the cases of $P_1$ and $P_2$ the Levi subgroup is not purely even, while for $P_1\cap P_2$ it is so.

 The decompositions written above for $\fp_1$, $\fp_2$ and $\fp$ are semidirect products. Generically
 $$\fp=\fg_0\oplus \fp_+,\qquad \fp_+=\bigoplus_{k>0}\fg_k$$ which are also called Levi decompositions, as in the non super case.

 \hfill$\blacksquare$

\begin{definition}
The \textit{complex conformal supergroup} is the
complex special linear supergroup $\rSL(4|1)$. The \textit{complex
Poincar\'{e} supergroup times dilations} is the  subgroup of $\rSL(4|1)$ given,
in the functor of points notation, by
$$
\mathrm{R}(A)=\left\{\begin{pmatrix}L&0&0\\NL&R&R\chi\\d\varphi&0&d\end{pmatrix}
\right\} \subset \rSL(4|1)(A).
$$

\hfill$\square$
\end{definition}

We can denote $P_u=P_1\cap P_2$, so  $\F= \rSL(4|1)/P_u$. $\F$ is an  homogeneous superspace \cite{fl,ccf} of $\rSL(4|1)$ that we can call the  {\it complex conformal superspace}.
As in the non super case, a point in $\F$ can be given in terms of the basis vectors of the corresponding subspaces of $\C^{4|1}$ ( the {\it  twistor superspace}). In the big cell the minor $d_{12}\neq 0$, so we can bring both basis to standard forms
$$\begin{CD}\left(\begin{pmatrix}\rid\\A\\\alpha\end{pmatrix}, \begin{pmatrix}\rid&0
\\B&\beta\\
0&1\end{pmatrix}\right)\end{CD},\qquad\hbox{with }\, B=A-\beta\alpha\,.$$ The last relation expresses the fact that the first space is a subspace of the second one in the big cell, and it is equivalent to the relations (\ref{incidence3}) locally, once the condition $d_{12}\neq 0$ is imposed.

 The action of the Poincar\'{e} supergroup times dilations is
$$\begin{CD}A@>>>R(A+\chi\alpha)L^{-1}+ N,\\
\alpha@>>>d(\alpha+\varphi)L^{-1},\\
\beta@>>>d^{-1}R(\beta+\chi)\,.\end{CD}$$

\subsection{The bundle of (super)conformal densities}

Embeddings of a variety into projective spaces are in one to one correspondence with a certain class of line bundles called {\it very ample line bundles}. These are bundles that have enough global sections to be used as projective coordinates of a   projective embedding (see for example Ref. \cite{skkt}).

 Theorem \ref{projectivetheorem} does not have an extension to the super setting. Nevertheless, embeddings into projective superspace are also determined by {\it very ample line superbundles}. The goal of the this section will be to describe explicitly a line superbundle associated to the projective embedding of the superflag $\F$ explained in detail in Section \ref{superflag}.

 \medskip

We consider first the non super case. The reader can resort to Ref. \cite{cs} for more details.

As before, let $G$ be an algebraic, semisimple Lie group and $H$ a closed subgroup with Lie algebras $\fg$ and $\fh$ respectively. Let $\pi:G\rightarrow G/H$ be the canonical projection. The sheaf of regular functions   $\mathcal{R}_{G/H}$ can be constructed in terms of $\cO_G$  and an invariance condition. Let  $U\subset_{\mathrm{open}} G/H$. Then  $\pi^{-1}(U)$ is invariant under the action of $P$. One can define
\be \mathcal{R}_{G/P}(U) =\{f\in \cO_G(\pi^{-1}(U))\; |\; f(gp)=f(g)\quad \forall g\in \pi^{-1}(U),\; p\in P\}\,.\label{homogeneoussheaf}\ee

If $H$ is a parabolic subgroup, $H=P$, then $G/P$ is a projective variety and the set of global regular functions $\mathcal{R}(G/P)=k$. In the following, we will describe a series of line bundles whose global sections reconstruct, degree by degree, the $\Z$-graded algebra $S=\oplus_{n=0}^\infty S^n $ that gives the scheme  $\uproj (S)$ of Example \ref{constructionproj}, associated to the projective variety $G/P$.

We consider now the Levi subgroup $G_0$ of $P$. Any representation of  $G_0$ can be extended to the whole $P$, by virtue of the Levi decomposition, assuming that it is trivial on $P_+$.
In particular, the group $G_0$ acts on $\fp_-:=\fg_{-k}\oplus \cdots \oplus \fg_{-1}$ with the adjoint representation. We shall use this representation to obtain  a representation of $G$ by the method of {\it parabolic induction}, which we briefly describe.

Let $\bV$ be  a $P$-module and consider the associated vector bundle
$G\times_P\bV$ over $G/P$. The space of global  sections of this bundle, $\Gamma(G\times_P\bV)$ is given by
$$\Gamma(G\times_P\bV)=\{f\otimes v\in  \cO(G)\otimes \bV\; |\; f(gp)\otimes v=f(g)\otimes p^{-1} v\quad \forall g\in G,\; p\in P\}\,.$$
On this space there is  a natural action of $G$.
A  class of $P$-modules can be obtained starting with a representation of $G_0$ and extending it trivially to the full $P$.
If $\bV\cong \C$, then the representation  is a character of $P$ and the bundle is a line bundle.
In particular, one can take the determinant (or a power of it) of the adjoint action of $G_0$ on $\fp_-$ :
$$
\begin{CD}
G_0@>\chi>>\C\\
g@>>> |\det(\Ad_-(g))|^{-\frac 1 d}\,.
\end{CD}
$$
In the above formula, $d$ is the dimension of $\fp_-$. We remark that the form of the exponent is purely conventional and carefully tuned for the purpose of this paper; when dealing with the Weyl structure it is natural to write it differently\footnote{We thank Andreas Cap and Rod Gover for this observation.}. We will denote the induced bundle of the  character $\chi$ as $\mathcal{L}$. Global sections of $\cL$ are $P$-equivariant functions:
\be
\Gamma(\mathcal{L})
=\{f: G \lra \mathbb{C}\; | \; f(gh)=\chi^{-1}(h)f(g)\}\,.\label{sectionsbundle}
\ee

We also denote  as $\mathcal{L}^n$ the bundle obtained by using the character $\chi^n$ (so $\cL=\cL^1$) and by $\Gamma(\cL^n)$ the set of its global sections. It is then natural to construct the graded algebra
\be \cL^*:=\bigoplus_{n \ge 0}\Gamma(\cL^n)\,,\label{gradedbundlealgebra}\ee which  is generated as an algebra in degree one.





We now come back to the particular example where $G=\rSL_4(\C)$ and $P=P_u$ . In this case, an easy computation shows that $\chi(g)=(\det L)^{-1}$, and we refer to the bundle $\mathcal{L}^n,$ constructed out of this character, as the {\it bundle of conformal densities of weight $n$}. We will also use Notation \ref{notationindices}, so $\det L= d_{12}(=d_{12}^{12})$, that is, the determinant of the upper left $2\times 2$ matrix in $\rSL_4(\C)$. This will be useful mainly to connect with the notation used in the quantum case.

\begin{remark} These bundles have an interesting interpretation that we are going to discuss. We start by noticing that
on the conformal space $\Gr(2,4)=\rSL_4(\C)/P_u$ there is no invariant Riemannian metric, but there is a more generic structure called a {\it conformal metric} (see for example Ref. \cite{va}, page 100 or Ref. \cite{fl}, page 101). At each point of the manifold we associate a set of  non singular, quadratic, symmetric, holomorphic (or smooth, real analytic,...) forms on its tangent space which are non zero scalar multiples of each other. We assume that the choice is such that on a neighborhood there is an holomorphic metric whose quadratic form belongs to that set, but such choice is perhaps not possible globally. For $\Gr(2,4)$, it is not difficult to prove that the Minkowski metric on the big cell defines an invariant conformal metric in that neighborhood and that changing from one open set to another amounts to multiply the Minkowski metric by a non zero factor. In the notation of (\ref{minkowski}), where the Minkowski space is $ M_2(\C)$, the quadratic form becomes
\begin{align*} q(A)&=\det A=ad-bc=(x^0)^2-(x^1)^2-(x^2)^2-(x^3)^2, \\ \\ A& =\begin{pmatrix}a&b\\c&d\end{pmatrix}=\begin{pmatrix}x^0+x^3&x^1-\ri x^2\\x^1+\ri x^2&x^0-x^3\end{pmatrix}\,.\end{align*}
A conformal metric defines a principal bundle  $\mathcal{Q}\subset T^*M\otimes T^*M $ with fiber $\C^\times=\C-\{0\}$ and local sections of this bundle are in one to one correspondence with a local choice of metric representing the conformal structure. Since the quadratic form in the big cell is just $\det A$, under the action of $G_0$ it transforms as
$$\begin{CD}\det A@>>>\det L^2\, \det A\,.\end{CD}$$ Associated to this principal bundle by the actions on $\C$
$$  \begin{CD}\C^\times\times \C@>>>\C\\
(\Omega^2, z)@>>> \Omega^n z\end{CD}$$
we reencounter all the bundles $\cL^n$.

In the conformal case, a section of $\cL$ is known as a {\it  conformal scale} since it amounts to give a local choice of units of length. This is strictly related to the the notion of {\it dilaton field} that is the gauge field of the dilations.

\hfill$\square$

\end{remark}

\begin{proposition}\label{amplebundle}
 The bundle of conformal densities of weight one  over \\$\rSL(4,\C)/P_u$, that is, the line bundle $$\cL=\rSL(4,\mathbb{C})\times_{P_u}\mathbb{C}=\cO(\rSL_4(\C)/P_u)_1$$ defined by the character $d_{12}$ of $P_u$   is very ample. \end{proposition}
\begin{proof}
 First we note that the determinants $\{d_{ij}(=d_{ij}^{12})\}$ are $P_u$-equivariant functions on $\rSL(4,\mathbb{C})$ as required in (\ref{sectionsbundle}):
$$ gp=\begin{pmatrix}A&B\\
C&D\end{pmatrix}\begin{pmatrix}L&Q\\
0&R\end{pmatrix}=\begin{pmatrix}AL&AQ+BR\\
CL&CQ+DR\end{pmatrix},\qquad g\in \rSL(4,\C),\; p\in P_u\,.$$
  Moreover,  they cannot be all zero at the same point and, at each point, they span the fiber of the line bundle. They are in fact the standard Pl\"{u}cker coordinates of the embedding of $G(2, 4)$ into $\bP^5$ (see for example Refs. \cite{va, fl}).
 \end{proof}

\medskip

We are now ready to generalize these structures to the super setting. Similarly to the classical case, we can establish a correspondence between certain  super line bundles on a supervariety $X$ and the embeddings of $X$ into projective superspaces. Let us start with a basic definition (see for example Ref \cite{ccf}, Section 10.5).

\begin{definition}
Let $S=(|S|, \cO_S)$ be a superscheme. A {\sl super vector bundle $\mathcal{V}$ of rank $p|q$ over $S$} is a locally free sheaf of $\cO_S$-modules of rank $p|q$. That is, for each $x \in |S|$ there exists an open set $U_x\subset |S|$
such that $\mathcal{V}(U) \cong \cO_S(U)^{p|q}:= \cO_S(U) \otimes k^{p|q}$.

The {\it  stalk at a point $x\in |S|$, $\mathcal{V}_x$} , is
the  $\cO_{S,x}$-module $\cV_x:=\cO_{S,x}\otimes k^{p|q}$,

The fiber over the point $x$ is the super vector space $\mathcal{V}_x/(m_x \mathcal{V}_x)\cong k^{p|q}$,  with $m_x$ being the maximal ideal of $\cO_{S,x}$.

 A {\sl super line bundle $\mathcal{V}$ on $S$}  is a rank $1|0$ super vector bundle over $S$.

\hfill$\square$
\end{definition}

Super line bundles can also be very ample, and they are related with  projective embeddings. It is instructive to understand the construction of the super line bundle in the simplest case, the projective superspace.

\begin{example}{\sl Very ample super line bundle on the projective superspace.}\label{veryampleps} Going back to Example \ref{constructionproj}, we consider $\uproj(A)$ with $$A=k[x_0, x_1, \dots, x_m; \xi_1,\dots ,\xi_n]\,.$$ One just takes the  graded $A$-module $A(1)$, defined by shifting he degree $A(1)^d=A^{d+1}$ and proceeds, by standard localization techniques, to construct the sheaf of modules over the scheme $\uproj(A)$. This is known as the {\it twisting sheaf of Serre} (see for example Ref. \cite{ha}). A basis of the space of global sections is then the set of   homogeneous coordinates $\{x_0, x_1, \dots, x_m; \xi_1,\dots ,\xi_n\}$.

\hfill $\square$

\end{example}

Let $G$ be a  super Lie group (see Example \ref{algebraicsupergroup}), in the algebraic or differential setting. Let $V$ be a super vector space. A representation of $G$ on $V$ is a morphism of super Lie groups\footnote{There are equivalent ways of seeing a representation of a super Lie group: Especially important is the construction in terms of super Harish-Chandra modules \cite{cfv}, but this definition will be enough for our purposes.}
$$\begin{CD}G@>\chi>> \rGL(V)\,.\end{CD}$$
 We will denote as $\chi_c$ the {\it contragradient representation}  on the dual space $V^*$.

As in the non super case, one can construct representations of $G$ on the space of sections of a certain super vector bundle, induced from a finite dimensional representation of a subgroup. Let $P\subset G$ be a closed subgroup\footnote{Later on, $P$ will be a parabolic subgroup.} and let $G_0$ and $P_0$ be the reduced groups of  $G$ and $P$. We assume that $G_0$ is connected.

We first describe the quotient superspace $G/P$ in terms of the sheaf of regular functions $G_0/P_0$. Let  $\pi:G\rightarrow G/P$ and $\pi_0:G_0\rightarrow G_0/P_0$ be the canonical projections. Let $\mathcal{R}_G$ be a sheaf obtained similarly to (\ref{homogeneoussheaf}): if $\mu:G\times G\rightarrow  G$ is the group multiplication (a map of sheafs), we can consider the composition
$$\mu_{G,P}:G\times P\xhookrightarrow{\id\times i}G\times G\xrightarrow{\mu}G$$ and, for  an open set   $V\subset G\times P$, the corresponding map of superalgebras $$\mu^*_{G,P}:\cO_G \left(\mu^{-1}(V)\right)\rightarrow \cO_{G\times P}(V)\,.$$
The quotient superspace is the topological space $G_0/P_0$ together with the sheaf
\be \mathcal{R}_{G/P}(U):=\{f\in \cO_G \left(\pi_0^{-1}(U)\right)\; | \; \mu_{G,P}^* f=f \}\,,\label{homspace}\ee where $U$ is an open set $U\subset G_0/P_0$.

Remember that for a supergroup, the functor of points is a group valued functor.
One can define a functor acting on superschemes  as $T\rightarrow G(T)/P(T)$. In general, it is not representable, but there exists always the {\it sheafification} of such functor which is so. The  functor is then the functor of points of the superscheme defined in (\ref{homspace}). The proof of this fact and further details can be found in  Section 9.3 of Ref. \cite{ccf}.

We now consider the sheaf obtained by tensoring  with the super vector space $\bV$:
$$\cA(U):=\cO_G \left(\pi_0^{-1}(U)\right)\otimes V\,.$$ If $\chi$ is a representation of $P$ on $V$, we  select the appropriate equivariant sections:
\be\cA_{\mathrm{inv}}(U)=\{\;f\in \cA(U)\; | \;\; (\mu_{G,P}^*\otimes 1)f=(1\otimes \chi^*_c) f\;\}\,.\label{equivariance}\ee  We have the following proposition \cite{cfv}:

\begin{proposition}
Let $G$ and $P$ be as above. Then the sheaf $\cA_{\mathrm{inv}}$ is a super vector bundle over $G/P$ with fiber $\bV$.

\hfill$\square$
\end{proposition}

It is useful  to write this in terms of the functor of points. Let $A$ be a  superalgebra and let $g$ be an $A$-point of $G$ and $p$ an $A$-point of $P$, so they are morphisms from $\cO(G)$ to $A$ (see Example \ref{algebraicsupergroup}).
We identify $V$ with the affine space  $k^{p|q}$, so let $v$ be an $A$-point of $V$ (see Example \ref{fop affine superspace}). We consider the set of global sections $\cO(G)$.
Let $f$ be an element of $\cO(G)\otimes V$. It is  an invariant element if
\be (gh\otimes v)(f)=(g \otimes\chi(h)^{-1}v)f\,. \label{globalinvariance}\ee

We now consider  $G=\rSL(4|1)$ and the parabolic subgroups associated to the super Grasmannians $\Gr_1, \Gr_2$ and the superflag $\F$.
Then,  $d_{12}$ is a character of both, $P_1$ and $P_u=P_1\cap P_2$. As we are going to see, it is only for $P_1$ that the induced super line bundle is very ample. So, in complete analogy with the classical case, we  construct over $\Gr_1$ the associated bundles to $\rSL(4|1)$ by using  the character $\chi(g)=d_{12}$ of $P_1$; we name it the {\it bundle of antichiral superconformal densities} . It is an interesting object due to the following
\begin{proposition}\label{vaslb}
 The bundle of antichiral superconformal densities of weight one  over $\Gr_1$, that is, the line bundle $$\rSL(4|1)\times_{P_1} \mathbb{C}$$ defined by the character $d_{12}$ of $P_1$   is very ample. \end{proposition}

\begin{proof}
It is enough to observe again that $\{d_{ij},d_{55},\delta_{i5}\}$ are $P_1$-equivariant sections on $\rSL(4|1)$ and they can not all vanish at the same time. In terms of the functor of points, this amounts to a simple check of the condition (\ref{globalinvariance}),  like in Proposition \ref{amplebundle}.
\end{proof}

We conclude this section with the following natural remark:
\begin{remark}  The very ample line bundle associated to the Pl\"{u}cker embedding of $\Gr_2$ can be constructed out of the character $d^*_{12}$; we call it the {\it bundle of chiral superconformal densities}, and it is obviously very ample too. It is a natural task then to construct the very ample line bundle associated to the Segre embedding;  we postpone this discussion to the next section where we will introduce the key concepts of {\it  classical} and {\it quantum sections}, and we will use the intuition arising from the bundles of chiral and antichiral superconformal densities to characterize the coordinate ring of the superflag. The full answer is in Example \ref{sft}.

\hfill$\square$
\end{remark}

\section{The quantum Grassmannians $\Gr_{1,q}$ and $\Gr_{2,q}$}\label{qgrass-sec}
We start with the
definition of the  quantum matrix superalgebra due to Manin \cite{ma2}.

We will denote $\C_q: =\C[q,q^{-1}]$.

\begin{definition}\label{qms}
The {\it quantum matrix superalgebra}  is given by
$$
M_q(m|n)=_{\mathrm{def}}\C_q\langle a_{ij}\rangle /I_M,\qquad i,j=1,\dots ,n\,,
$$
where $\C_q\langle a_{ij}\rangle$ denotes the free algebra over $\C_q=\C[q,q^{-1}]$
generated by the homogeneous variables $a_{ij}$
and the ideal $I_M$ is generated by the relations \cite{ma2}:
\begin{align*}
&a_{ij}a_{il}=(-1)^{\pi(a_{ij})\pi(a_{il})}
q^{(-1)^{p(i)+1}}a_{il}a_{ij}, \quad &&j < l \\ \\
&a_{ij}a_{kj}=(-1)^{\pi(a_{ij})\pi(a_{kj})}
q^{(-1)^{p(j)+1}}a_{kj}a_{ij}, \quad &&i < k \\ \\
&a_{ij}a_{kl}=(-1)^{\pi(a_{ij})\pi(a_{kl})}a_{kl}a_{ij}, \quad
&&i< k,j > l \quad \hbox{or} \quad i > k,j < l \\ \\
&a_{ij}a_{kl}-(-1)^{\pi(a_{ij})\pi(a_{kl})}a_{kl}a_{ij}=
\eta(q^{-1}-q)a_{kj}a_{il} \quad &&i<k,j<l
\end{align*}
where
\begin{align*}
&i,j,k,l=1, \dots m+n,\qquad
\eta=(-1)^{p(k)p(l)+p(j)p(l)+p(k)p(j)}, \\
&p(i)=0\; \hbox{ if } \quad  1 \leq i \leq m, \qquad  p(i)=1\; \hbox{ if } \quad m+ 1 \leq i \leq n+m\quad \hbox{ and }\quad
\\&\pi(a_{ij})=p(i)+p(j)\,.
\end{align*}
\hfill$\square$
\end{definition}
It is usual to organize the generators $a_{ij}$ in matrix form as $$M= (a_{ij})= \begin{pmatrix}A_{m\times m}&\Lambda_{m\times n}\\\Upsilon_{n\times m}&B_{n\times n}\end{pmatrix}\,,$$ where the dimensions of the blocks are indicated. One says that a matrix like $M$ is a {\it quantum supermatrix}.
$M_q(m|n)$ is a bialgebra with comultiplication and
counit given by
\be
\Delta(a_{ij})=\sum_k a_{ik} \otimes a_{kj},
\qquad
\ep(a_{ij})=\de_{ij}\,.
\label{coproduct}
\ee
We observe that the comultiplication can be formally understood as matrix multiplication.

If $n=0$ one reduces to the even case and the relations above are the ones of the standard {\it quantum matrices} $M_q(m)$. The standard quantum group $\rGL_q(m)$ is then obtained by inverting the determinant, that is,
$$\rGL_q(m)=_{\mathrm{def}}
M_q(m)
\langle D^{-1}\rangle\,,
$$
where $D^{-1}$ is an (even) indeterminate such
that
$$
DD^{-1}=1=
D^{-1}D\,.$$
$\rGL_q(m)$ is a Hopf algebra with antipode given by:
$$S_e(A_{ij})=(-q)^{i-j} {\det}_q A(j,i)\, {\det}_q A^{-1}\,, $$
where $A(j,i)$ is the quantum  matrix obtained from $A$ removing the $j$-th row and the $i$-th column and ${\det}_q$ is the quantum determinant:
$${\det}_q A =_{\mathrm{def}}\sum_{\s \in S_m}(-q)^{-l(\s)}
a_{1\s(1)} \cdots a_{m\s(m)}\,. $$
\begin{definition}
The {\it quantum general linear supergroup} is defined as
$$
\rGL_q(m|n)=_{\mathrm{def}}
M_q(m|n)
\langle{D_{1}}^{-1},{D_{2}}^{-1}\rangle\,,
$$
where ${D_{1}}^{-1}$ and
${D_{2}}^{-1}$ are even indeterminates such
that
$$
{D_1}D_1^{-1}=1=
{D_1}^{-1}D_{1},\qquad
{D_{2}}
{D_{2}}^{-1}=1=
{D_{2}}^{-1}
{D_{2}}\,,
$$
and
$$
D_{1}=_{\mathrm{def}}{\det}_q A,\qquad
D_{2}=_{\mathrm{def}}{\det}_q B\,,
$$
are the quantum determinants of the diagonal blocks.

\hfill$\square$
\end{definition}
$\rGL_q(m|n)$ is a Hopf algebra with the comultiplication and
counit  as in $M_q(m|n)$ and the antipode $S$  \cite{ph}  is given (in matrix form) as follows:
\begin{align}
&S(A)=
S_e(A)+S_e(A)\left(\Lambda S_e(H)\Upsilon \right)S_e(A), &&
S(B)=S_e(H), \nonumber\\
&S(\Lambda)=-S_e(A)\Lambda S_e(H)
&&S(\Upsilon)=-S_e(H) \Upsilon S_e(A),
\nonumber\\
& S(D_1^{-1})= D_1,   &&S(D_2^{-1})=D_2\,, \label{antipode}
\end{align}
where $H:=B-\Upsilon
S_e(A)\Lambda$.

One can compute the commutation relations for the entries of $H$ and check that it is a quantum
matrix (see Ref. \cite{ph}, Section 4) thus $S_e(H)$ is well defined.

Manin \cite{ma2} also introduced the {\it quantum berezinian} (see also Refs. \cite{ph,fi5,zhang}), which is a {\it central} and {\it group like} element
$$\ber_q M=_\mathrm{def}{\det}_qA\,{\det}_{q}\left(S_e(B-\Upsilon
S_e(A)\Lambda)\right)\,.$$

The {\it quantum special linear  supergroup}
is defined as
$$\rSL_q(m|n)=_\mathrm{def}\rGL_q(m|n)/\langle B_q-1\rangle\,. $$
In  $\rSL_q(m|n)$ (and $\rGL_q(m|n)$ or $M_q(m|n)$))
we can define an algebra-morphism, the {\it supertranspose}:
$$\begin{CD}\rSL_q(m|n)@>
\mathrm{st}>> \rSL_q(m|n) \\ a_{ij}@>>>
\mathrm{st}(a_{ij})=(-1)^{(p(j)+1)p(i)}a_{ji}\,.\end{CD}
$$

As one can readily check, the Manin relations and $\langle B_q-1\rangle$ are invariant under  $\mathrm{st}$. Then, the supertranspose  is an algebra automorphism of  $\rGL_q(m|n)$ and $\rSL_q(m|n)$. It is also immediate to check that with respect to the coalgebra structure, it is an antiautomorphism.

 The following proposition,
 will turn out to be very important for us.

\begin{proposition} \label{iso-St}
Let $S$ be the antipode in $\rGL_q(m|n)$.
The map $S \circ \mathrm{st}$ is a Hopf algebra isomorphism of the two quantum
superalgebras $\rSL_q(m|n)\rightarrow \rSL_{q^{-1}}(m|n)$ (where the
latter is the same as
 $\rSL_q(m|n)$, but with $q$ replaced by $q^{-1}$).
\end{proposition}

\begin{proof}
The antipode is always an antiautomorphism of the Hopf algebra,
so the composition $S \circ \mathrm{st}$ is an algebra antiautomorphism and a coalgebra automorphism of $\rSL_q(m|n)$.
 The map
$$ a_{ij}\rightarrow a_{ij},\qquad q\rightarrow q^{-1}\,, $$  can be extended to an algebra-antiautomorphism. With respect to the coalgebra structure, it is an automorphism.

Then, $S \circ \mathrm{st}$ can be seen as an isomorphism $\rSL_q(m|n)\rightarrow \rSL_{q^{-1}}(m|n)$.

\end{proof}

We now go back to our special case, and we consider $\rSL_q(4|1)$.
This is a quantum deformation of $\rSL(4|1)$, hence we call it
the {\it quantum conformal supergroup}.

Different parabolic subgroups as ${P_1}_q$, ${P_2}_q$ and ${P_u}_q=(P_1\cap P_2)_q$ are defined as quotients of this superalgebra by the an ideal. Formally, this ideal is the same than in the classical case, that is, the ideal of the entries that are put to zero in (\ref{parabolics}).

The superflag $\F$ can be realized both, as a suitable quotient of
$\rSL(4|1)$
and  as embedded the product of two super
Grassmannians, $\Gr_1\times\Gr_2$. The first point of view gives
immediately the action of the conformal supergroup and the realization
of $\F$ as homogeneous superspace, while the second one allows to associate
to $\F$ the coordinate superalgebra of Theorem \ref{embflag}.

In Section \ref{superflag} we have
provided a detailed description of the coordinate superalgebras of $\Gr_1$
and $\Gr_2$. We now want to give
a quantization of them  \cite{cfl,fl}. We will use them
to obtain a quantum deformation of the superflag $\F$.

The following definition is motivated by the fact (proven at the end of Section \ref{pluckergrass}) that $\C[\Gr_1]$ and $\C[\Gr_2]$ can be seen as subalgebras of $\C[\rSL(4|1)$.

\begin{definition}\label{qsg}
Let the notation be as above. We define the {\it quantum super Grassmannians}
$\Gr_{1,q}$, $\Gr_{2,q}$ as the following $\Z$-graded subalgebras defined inside $\rSL_q(4|1)$
(equivalently, inside $\rGL_q(4|1)$).
The superalgebra $\Gr_{1,q}$ is generated by the following
 quantum super minors:
\begin{align*}
&D_{ij}= a_{i1}a_{j2}-q^{-1}a_{i2}a_{j1}, \quad &&1\leq i<j \leq 4,\\
&D_{i5}= a_{i1}a_{52}-q^{-1}a_{i2}a_{51}, \quad &&1 \leq i \leq 4,\\
&D_{55}=a_{51}a_{52}\,, &&
\end{align*}
while  $\Gr_{2,q}$ is generated by

\begin{align*}
&D_{ij}^*= a^{i3}a^{j4}-qa^{i4}a^{j3}, \quad &&1\leq i<j \leq 4, \qquad
 \\
&D_{i5}^*= a^{i3}a^{54}-qa^{i4}a^{53},
\quad &&1 \leq i \leq 4,\\&D_{55}^*=a^{53}a^{54}\,,&&
\end{align*}

where we have written, as usual,  $a^{ij}=(S\circ \mathrm{st})(a_{ij})$.

\hfill$\square$

\end{definition}
The following proposition gives the true meaning of these deformations.

\begin{proposition}\label{pres-Gr}
The generators of the subalgebra $\Gr_{1,q}$
satisfy the following relations, which provide
a presentation:

\smallskip

{\it Quantum super Pl\"{u}cker relations.}
\begin{align*}
&D_{12}D_{34}-q^{-1}D_{13}D_{24}+q^{-2}D_{14}D_{23}=0, \\
&D_{ij}D_{k5}-q^{-1}D_{ik}D_{j5}+q^{-2}D_{i5}D_{jk}=0,
\qquad &1 \leq i<j<k \leq 4 ,\\
&D_{i5}D_{j5}=qD_{ij}D_{55}, \qquad &1 \leq i<j \leq 4\,.
\end{align*}

{\it Commutation relations.}
\begin{itemize}
\item If $i,j,k,l$ are {\it not}
all distinct and $D_{ij}$, $D_{kl}$ are not both odd, we have:
\begin{equation}
D_{ij}D_{kl}=q^{-1}D_{kl}D_{ij}, \quad (i,j)<(k,l),
\quad 1 \leq  i,j,k,l \leq 5, 
\end{equation}
where `$<$'     refers to the lexicographic ordering.

\item
If $i,j,k,l$ are  all distinct  and $D_{ij}$, $D_{kl}$ are not both odd, we have:
\begin{align*}
 &D_{ij}D_{kl}=q^{-2}D_{kl}D_{ij},   \qquad &&  1 \leq i<j<k<l\leq 5,
  \\
 &D_{ij}D_{kl}=q^{-2}D_{kl}D_{ij}-(q^{-1}-q)D_{ik}D_{jl},\qquad &&
   1 \leq i<k<j< l \leq 5,   \\
& D_{ij}D_{kl}=D_{kl}D_{ij},  \qquad &&  1 \leq i<k<l<j\leq 5\, .
\end{align*}
\item Commutations with $D_{55}$ or involving two odd elements:
\begin{align*}
& D_{ij}D_{55}=q^{-2}D_{55}D_{ij}, \\
 &D_{i5}D_{j5}=-q^{-1}D_{j5}D_{i5}-(q^{-1}-q)D_{ij}D_{55}=
-qD_{j5}D_{i5} \\
 &D_{i5}D_{55}=D_{55}D_{i5}=0\,.
\end{align*}
\end{itemize}

The subalgebra $\Gr_{2,q}$ admits the same presentation where
$q$ is replaced by $q^{-1}$ and $D$ is replaced with $D^*$.
\end{proposition}

\begin{proof} The claim about the presentation of $\Gr_{1,q}$ is
proved in Chapter 5 of Ref. \cite{fl}, while the one about $\Gr_{2,q}$ is an
immediate consequence of Proposition \ref{iso-St}.
\end{proof}

 As stated in Refs. \cite{cfl,fl} one can prove that $\Gr_{1,q}$ and $\Gr_{2,q}$ are quantum homogeneous spaces
for the quantum supergroup $\rSL_q(4|1)$ (also for $\rGL_q(4|1)$). This is the content of the next proposition:

\begin{proposition}
There is a well defined coaction of the quantum supergroups
$\rGL_q(4|1)$ and $\rSL_q(4|1)$ on $\Gr_{1,q}$ and $\Gr_{2,q}$, obtained
by restricting the comultiplication.
On the generators, such coaction is given explicitly by the formulas
$$
\Delta(D_{ij})= \sum D_{ij}^{kl} \otimes D_{kl}\,, \qquad
\Delta(D_{ij}^*)= \sum {D^*}_{ij}^{kl} \otimes D_{kl}^*\,.
$$
\end{proposition}
\begin{proof}
Direct calculation, essentially the same as in Proposition 1.4 of Ref. \cite{fi1}.
\end{proof}

At this point one may be tempted, in analogy with the ordinary
setting (see Remark \ref{remarksubring}),
to define the quantum superflag as the quantum
subsuperalgebra of $\rSL_q(4|1)$ generated by the elements
$D_{ij}$, $D^*_{kl}$ defined above.  Such definition  would
require us to compute the commutation relations of any
pair $D_{ij}$, $D_{kl}^*$, in order to make sure
that this subsuperalgebra is well defined. In other words, one has to prove that no other
elements besides $D_{ij}$ and $D^*_{kl}$
appear actually in the commutation relations. Moreover, in order to give a presentation similar to Theorem \ref{embflag}, one would have to compute a generalization of the incidence relations (\ref{incidence3}).
As Proposition \ref{pres-Gr} shows, these
relations are highly non trivial to compute and for this reason
we prefer to take another route.

We will define the quantum superflag
via the notion of a {\it quantum section} \cite{fi6} of the very ample super line bundle related to the projective embedding of the conformal superspace. This will allow us to give a characterization of the quantum coordinate ring.

\section{The quantum section}\label{qsection-sec}

The global sections of the super line bundle  $\cL$ are
characterized by the equivariance condition (\ref{globalinvariance}).
We can express it in pure Hopf algebraic terms. Since $\bV =\C$ in our case, we identify $\cO(G)\otimes \bV\cong \cO(G)$.  Let $I(P)$ the ideal in $\cO(G)$ defining $P$, so $\cO(P)=\cO(G)/I(P)$,  and denote as  $\pi: \cO(G) \rightarrow \cO(P)$ the canonical projection. Let  $\Delta:\cO(G)\rightarrow \cO(G)\otimes \cO(G)$ be the coproduct in $\cO(G)$.  Then we have
$$
\cO(G/P)_1=\Big\{\, f \in \cO(G) \,\;\Big|\;
(\id \otimes \pi)\Delta(f) = f \otimes S(\chi) \Big\}\,.
$$
 Let $t\in \cO(G)$ such that $t=\pi(\chi)$.  If $\cL$ is very ample --it corresponds
to a projective embedding--  we have the following important result
\cite{fi6}:

\begin{proposition} \label{t}
Let the notation be as above. Let the supervariety
$G/P$ be embedded into some projective superspace via the line
bundle $\cL$. Let $\pi: \cO(G) \rightarrow \cO(P)=\cO(G)/I(P)$ and $\Delta $ the coproduct in $\cO(G)$ (formally the same as in (\ref{coproduct})).
Then, there exists an element $t \in  \cO(G)$, with $\pi(t)=\chi$,  such
that
\begin{align*}
&\left((\id \otimes \pi)
\circ \Delta \right)(t) \, = \, t \otimes \pi(t),   \quad
\pi\big(t^m\big) \not= \pi\big(t^n\big) \quad \forall
m \not= n \in \N\,,  \\ \\
&\cO(G/P)_n  \; = \;  \Big\{\, f \in  \cO(G) \,\;\Big|\;
(\id \otimes \pi)\Delta(f) = f \otimes \pi\big(t^n\big) \Big\}\,,  \\ \\
&\cO(G/P)  \; = \;  {\textstyle \bigoplus_{n \in \N}} \; \cO(G)_n\,,
\end{align*}
and $ \, \cO(G/P) \, $  is generated in degree 1, namely by
$ \cO(G/P)_1 \, $.

\hfill$\square$
\end{proposition}
We call $t$  the {\it classical section} associated to the super line bundle $\cL$. The following are the relevant examples.

\begin{example}\label{grt}For the Grassmannians $\Gr_1$ and $\Gr_2$, it is a calculation to show that the first condition is satisfied for  the elements $d_{12}\in \cO(\rSL(4|1))$ and $d^*_{12}\in \cO(\rSL(4|1))$.
The remaining conditions rely on Proposition \ref{vaslb}.

\hfill$\square$

\end{example}

\begin{example}\label{sft}

For the superflag, having  in mind the super Segre embedding, the natural guess for the classical section would be $t=d_{12}d_{12}^* \in \cO(\rSL(4|1))$. This is in fact true: one can check that the coordinates of the super Segre embedding (see Section \ref{segre-sec})
$$\left(\begin{array}{cc}
d_Id^*_K&d_{55}d^*_K\\
d_Id^*_{55}&d_{55}d^*_{55}\\
\hline
\delta_{i5}d^*_K&\delta_{i5}d^*_{55}\end{array}\vline
\begin{array}{c}
d_I\delta^*_k\\
d_{55}\delta^*_k\\
\hline
\delta_{i5}\delta^*_{k5}
\end{array}\right), \qquad \begin{array}{lcccc}I, K=&(1,2), &(1,3),  &(1,4), \\&(2,3), &(2,4), &(3,4)\,,\end{array}
$$
with $d$, $d^*$, $\delta$ and $\delta^*$ being the determinants defined at the end of Section \ref{pluckergrass}, are $t$-equivariant sections. We have then achieved a description of the coordinate ring of the projective embedding of $\F$ in $\bP^{64|56}$ as a (graded) subring of $\cO(\rSL(4|1)$. 

\hfill$\square$

\end{example}

We are now ready to transfer to the
quantum supergroup setting the notion of super line bundle and
the equivalent and corresponding notion of super projective
embedding (see Refs. \cite{fi6} and \cite{ccf} Ch. 10 for the ordinary setting).

We start
with the definition of the {\it quantum section}.
Let $\cO_q(G)$, $\cO_q(P)$ denote the  quantizations of the
superalgebras $ \cO(G)$ and $ \cO(P)$. Let $I_q(P)$ the ideal in $\cO_q(G)$ such that
$$\cO_q(P)=\cO_q(G)/I_q(P)\,.$$ We denote as
$\pi: \cO_q(G) \lra \cO_q(P)$ (no risk of confusion)

the canonical projection.

\begin{definition} \label{qsec}
Let $\cL$ be the super line bundle on  $ G / P $
given by  the classical section $ t$. A  \textit{quantum section} or quantization of $t$ is
 an element $ \, d \in \cO_q(G) \, $  such that

\begin{enumerate}
\item $(\id\otimes \pi) \Delta(d) =  d \otimes \pi(d) $.

\item $ t\,= \,d \mod (q\!-\!1) \, \cO_q(G)
$
\end{enumerate}

\end{definition}

Since $d \in \cO_q(G)$ reduces to $t$  when we specialize $q=1$, and $t$
contains all of the information to reconstruct the line bundle $\cL$,
we may think of $d$ as a quantum deformation of the line bundle $\cL$. Also,
$\cL$ corresponds to an embedding
of $G/P$ into a projective superspace and consequently a $\Z$-graded
superalgebra
$$\cO(G/P)=\sum_{n=0}^\infty \cO(G/P)_n$$ similar to the construction
in (\ref{gradedbundlealgebra}). We now use the quantum section $d$
to translate it the quantum case.

\begin{definition}
Let $d$ be a quantum section of $\cL$.
We define \begin{align*}
\cO_q(G/P)&:=\oplus_{n\in \N}\cO_q(G/P)_n, \\[3mm] &\hbox{where}\\[3mm]
\cO_q(G/P)_n &:= \bigl\{\,f \in \cO_q(G)\, | \,
(\id \otimes \pi)\Delta(f) = f \otimes \pi(d^n)\,\bigr\}\,.
\end{align*}

\hfill$\square$
\end{definition}

The next proposition is proven in Ref \cite{fi6} and it shows the
importance of quantum sections. We recall here that the superalgebra
$\cO_q(G/P)$ is seen as a subalgebra of $\cO_q(G)$, as in Definition \ref{qsg}.

\begin{theorem}  \label{Oqgh-graded}
Let  $d$  be a quantum section on  $ \, G\big/P \, $.  Then we have:
\begin{enumerate}
\item For all  $ \, r, s \in \N$
$$
\Oqgh_r \cdot \Oqgh_s \, \subseteq \, \Oqgh_{r+s}\,.
$$
Furthermore,
$$
\Oqgh ={\bigoplus}_{n \in \N} \Oqgh_n \, \subset \, \cO_q(G)\,.
$$

 \item The grading in  {\it (1)}  is compatible with the quantum
homogeneous space structure, that is,    $ \, \Oqgh $  is a  {\it graded}
$ \, \cO_q(G) $--comodule  algebra, via
the restriction of the comultiplication $\Delta$
in $\cO_q(G)$,  where we take on  $ \, \cO_q(G) $
the trivial grading:
$$
\Delta|_{\Oqgh}: \Oqgh \lra \cO_q(G) \otimes
\Oqgh
$$
\item For every  $ \, c \in \bk_q \, $,  we have
$ \;\; \Oqgh \, {\textstyle \bigcap} \, c \, \cO_q(G) \, = \,
c \, \Oqgh \,.$
  In particular,
$$
\Oqgh \, {\textstyle \bigcap} \,
(q\!-\!1) \, \cO_q(G) \, = \, (q\!-\!1) \, \Oqgh \,.
$$

\end{enumerate}
Hence $\Oqgh$ is a projective homogeneous quantum supervariety
for the coaction of the quantum supergroup $\cO_q(G)$.

\end{theorem}
\section{The quantum superflag.} \label{qsuperflag-sec}

We now  take $G=\rSL(4|1)$
and $P=P_u$  the upper parabolic subgroup of $\F$ (\ref{parabolics}):
$$
P_u(A)=
\left\{\left(
\begin{pmatrix}
P & Q & \nu \\
0 & R & 0 \\
0 & \beta & s
\end{pmatrix} \right) \right\} \subset \rSL(4|1)(A)\,,
$$

We intend to give the quantum deformation of the conformal superspace through a quantum section.  The superflag $\F$ is seen inside the product of the super Grassmanians, while its projective embedding is realized by means of the super Segre map. We already observed in the Example \ref{grt} that the Pl\"{u}cker embeddings for $\Gr_1$ and $\Gr_2$  are related to the classical sections  $d_{12}$ and $d_{12}^*$ respectively. For the flag, one has to consider the  classical section $t=d_{12}d_{12}^*$. We want now to construct
a quantum section $d$ which reduces modulo $q-1$ to  $t$.


We define as before:
$$
D_{12}=a_{11}a_{22}-q^{-1}a_{12}a_{21}, \qquad
D_{12}^*=a^{13}a^{24}-qa^{23}a^{14}
$$


The next statements are the main results of this section and
give a quantization of the conformal superspace.

\begin{proposition}\label{qsec-prop}
The element $d =D_{12}D_{12}^* \in \rSL_q(4|1)$ is a quantum section, with
respect to the super line bundle $\cL$ on $\rSL(4|1)/P_u$ given by
$t=d_{12}d_{12}^*$.
\end{proposition}

\begin{proof}
By Prop. 1.4 in Ref. \cite{fi1},we have that:
$$
\Delta(D_{12})=\sum_{1 \leq k<l \leq 5} D_{12}^{kl} \otimes D_{kl}\,.
$$
Hence, since $\pi(D_{kl})=0$, unless $(k,l)=(1,2)$, we have:
$$
(\id\otimes \pi)\Delta(D_{12})=D_{12} \otimes \pi(D_{12}).
$$
By \ref{iso-St} (see also \cite{zhang} Sec. 2) we also have that
$$
\Delta(a^{ij})=\sum a^{ik} \otimes a^{kj}
$$
hence, repeating a calculation similar to the one in Prop.
1.4 \cite{fi1} one obtains:
$$
\Delta(D_{12}^*)=\sum_{1 \leq k<l \leq 5}
{D^*}_{12}^{\,kl} \otimes D_{kl}^*\,.
$$
Since, as above,  $\pi(D_{kl}^*)=0$, unless $(k,l)=(1,2)$, we have
$$
(\id\otimes \pi)\Delta(D_{12}^*)=D_{12}^* \otimes \pi(D_{12}^*), \qquad
$$
Since $\Delta$ is multiplicative, i.e. $\Delta(D_{12}D_{12}^*)=
\Delta(D_{12})\Delta(D_{12}^*)$, we have our result.

\end{proof}

\begin{corollary}
The $\Z$-graded
subalgebra
$$
C_q:=\cO_q(G/P) \subset \rSL_q(4|1), \qquad G=\rSL(4|1), \, P=P_{u}
$$
defined by the quantum section $d =D_{12}D_{12}^*$ is a quantum deformation
of the graded subalgebra
of $\rSL_q(4|1)$ obtained via the classical section $t=d_{12}d^*_{12}$.

%
Furthermore $C_q$  has a natural coaction of the supergroup $\rSL_q(4|1)$.
Therefore it is a quantum homogeneous superspace.
\end{corollary}

\begin{proof} Immediate from Props. \ref{qsec-prop} and \ref{Oqgh-graded}.
\end{proof}


We can then call $C_q$ the {\it quantum conformal superspace}, because
it is a quantum deformation of  $\cO_q(\rSL(4|1)/P_u)$, the graded
algebra of the conformal superspace, with
respect to the Segre embedding discussed above.




\section*{Acknowledgements}

 R. Fioresi and E. Latini want to thank
the Departament de F\'{\i}sica Te\`{o}rica,
Universitat de Val\`{e}ncia for the warm hospitality
during the elaboration of this work.

M. A. Lled\'{o} would like to thank the
Dipartimento di Matematica, Universit\`{a} di Bologna for its kind hospitality during the realization of this work.

This work has been supported in part by grants  FIS2011-29813-C02-02, FIS2014-57387-C3-1 and  SEV-2014-0398 of the Ministerio de Econom\'{\i}a y Competitividad (Spain).

\appendix
\section{Incidence relations}\label{incidence-ap}
\begin{align*}
&-d_{12}d_{12}^*-d_{13}d_{13}^*-d_{14}d_{14}^*-\dd_{15}\dd_{15}^*=0, \\&
-d_{13}d_{23}^*-d_{14}d_{24}^*-\dd_{15}\dd_{25}^*=0, \\&
-d_{23}d_{13}^*-d_{24}d_{14}^*-\dd_{25}\dd_{15}^*=0, \\&
 \phantom{-}\,\,d_{12}d_{23}^*-d_{14}d_{34}^*-\dd_{15}\dd_{35}^*=0, \\&
 \phantom{-}\,\,d_{23}d_{12}^*-d_{34}d_{14}^*-\dd_{35}\dd_{15}^*=0, \\&
-d_{12}d_{24}^*+d_{13}d_{34}^*-\dd_{15}\dd_{45}^*=0, \\&
 \phantom{-}\,\,\dd_{15}\dd_{25}^*+d_{13}\dd_{35}^*+d_{14}\dd_{45}^*+\dd_{15}d_{55}^*=0,\\&
 \phantom{-}\,\,d_{12}d_{12}^*+d_{23}d_{23}^*+d_{24}d_{24}^*+\dd_{25}\dd_{25}^*=0, \\&
 \phantom{-}\,\,d_{12}d_{13}^*+d_{24}d_{34}^*+\dd_{25}\dd_{35}^*=0, \\&
 \phantom{-}\,\,d_{13}d_{12}^*+d_{34}d_{24}^*+\dd_{35}\dd_{25}^*=0, \\&
-d_{12}d_{14}^*+d_{23}d_{34}^*-\dd_{25}\dd_{45}^*=0, \\&
-d_{12}\dd_{15}^*+d_{23}\dd_{35}^*+d_{24}\dd_{45}^*+
\dd_{25}d_{55}^*=0, \\&
-\dd_{15}d_{12}^*+\dd_{35}d_{23}^*+\dd_{45}d_{24}^*-d_{55}\dd_{25}^*=0, \\&
\phantom{-}\,\,
d_{13}d_{13}^*+d_{23}d_{23}^*+d_{34}d_{34}^*+\dd_{35}\dd_{35}^*=0, \\&
-d_{13}d_{14}^*-d_{23}d_{24}^*-\dd_{35}\dd_{45}^*=0, \\&
-d_{14}d_{13}^*-d_{24}d_{23}^*-\dd_{45}\dd_{35}^*=0, \\&
-d_{13}\dd_{15}^*-d_{23}d_{25}^*+d_{34}\dd_{45}^*-\dd_{35}d_{55}^*=0, \\&
-\dd_{15}d_{13}^*-\dd_{25}d_{23}^*+\dd_{45}d_{34}^*-d_{55}\dd_{35}^*=0, \\&
-d_{14}d_{15}^*-d_{24}\dd_{25}^*-d_{34}\dd_{35}^*+\dd_{45}d_{55}^*=0, \\&
-\dd_{15}d_{14}^*-\dd_{25}d_{24}^*-\dd_{35}d_{34}^*+d_{55}d_{34}^*=0, \\&
-\dd_{15}\dd_{15}^*-\dd_{25}\dd_{25}^*-\dd_{35}\dd_{35}^*
-\dd_{45}\dd_{45}^*+d_{55}d_{55}^*=0, \\&
 \phantom{-}\,\,d_{24}d_{12}^*-d_{34}d_{13}^*-\dd_{45}\dd_{15}^*=0, \\&
 \phantom{-}\,\,\dd_{25}d_{12}^*+\dd_{35}d_{13}^*+\dd_{45}d_{14}^*-d_{55}\dd_{15}^*=0, \\&
-d_{14}d_{12}^*+d_{34}d_{23}^*-\dd_{45}\dd_{25}^*=0, \\&
 \phantom{-}\,\,d_{14}d_{14}^*+d_{24}d_{24}^*+d_{34}d_{34}^*-\dd_{45}\dd_{45}^*=0\,.
\end{align*}

\section{The super Segre map is an embedding}\label{embedding-ap}

We want to prove that the super Segre map (\ref{supersegre}) is an embedding. We will proceed by using the {\it even rules principle} by Deligne and Morgan \cite{dm}.

\begin{theorem} \label{erp}{\sl Even rules principle.} Let $\{V_i\}_{i\in I}$, $I=1, \dots , n$ be a family of super vector spaces, $V$ another super vector space and  $A=A_0\oplus A_1$ a commutative superalgebra. We denote
$ V_{i\,0}(A)=(A\otimes V_i)_0$ and $V_0(A)=(A\otimes V)_0$.

Any family of $A_0$-multilinear maps
$$\begin{CD}V_{1\,0}(A)\times \cdots \times V_{n\,0}(A)@>f_A>> V_0(A)\end{CD}$$
which is {\it functorial} in $A$ comes from a unique  morphism
$$\begin{CD}V_1\otimes \cdots\otimes V_n@>f>>  V\end{CD}\,,$$ that is,
$$f_A(b_1\otimes v_1, b_2\otimes v_2,\dots, b_n\otimes v_n)=(-1)^pb_1\cdots b_n\, f(v_1\otimes\cdots \otimes v_n)\,, $$ where $p$ is the number of pairs $(i,j)$ with $i<j$ and $v_i, v_j$ odd.

\hfill$\square$

\end{theorem}

Let $A$ be a local superalgebra. We consider the super vector spaces $\C^{n+1|r}$, $\C^{d+1|s}$ and its tensor product
$\C^{M+1|N}$ with $M+1=(n+1)(d+1)+rs$, $N=(n+1)s+(d+1)r$.
We recall the notation $A^{p|q}=A\otimes \C^{p|q}$ and we will also denote

\begin{align*}
&(x, \alpha):=(x_0,\dots , x_n\,|\alpha_1,\dots, \alpha_r)\in \C^{n+1|r},\\&
(y, \beta):=(y_0,\dots , y_d\,|\beta_1,\dots, \beta_s)\in \C^{d+1|s}\,.
\end{align*}

For shortness, if there is no possibility of confusion, we will denote
\begin{align}
&a\otimes x:=\sum_{i=0}^na_i\otimes x_i,\qquad \theta\otimes \alpha:= \sum_{a=1}^r\theta_a\otimes \alpha_a
\qquad a_i\in A_0,\; \theta_a\in A_1,\nonumber\\
&b\otimes y:=\sum_{i=0}^nb_i\otimes y_i,\qquad \xi\otimes \beta:= \sum_{a=1}^r\xi_a\otimes \beta_a
\qquad b_i\in A_0,\; \xi_a\in A_1\,.\label{vectors}
\end{align}
 There is a family of $A_0$-bilinear maps
\be \begin{CD}A^{n+1|r}_0\times A^{d+1|s}_0@>f_A>>  A^{M+1|N}_0 \\
(a\otimes x+\theta\otimes \alpha, b\otimes y+\xi\otimes \beta)@>>>ab\otimes (x\otimes y)-\theta\xi\otimes (\alpha\otimes \beta) \\@.+b\theta\otimes(\alpha\otimes y)+a\xi\otimes (x\otimes \beta)\,,\end{CD}\label{identitymaps}\ee
which is functorial in $A$. According to the even rules principle, there is a unique morphism
$$\begin{CD}\C^{n+1|r}\otimes\C^{d+1|s}@>f>>\C^{M+1|N}\end{CD}$$ which, in this case, is just the identity, $f=\id$. Although trivial, this morphism will help us to keep track of the signs.

In terms of the canonical basis of $\C^{n+1|r}$, $\C^{d+1|s}$ and its tensor product, we can represent an element of $A^{M+1|N}$ as a supermatrix with entries in $A$:
\be M_A=\begin{pmatrix} Z_A&\Lambda_A\\
\Gamma_A&T_A\end{pmatrix}\label{matrix}\,,\ee where  $Z_A$ is an $(n+1) \times(d+1)$ block,  $T_A$ is an $r\times  s$ block, $\Lambda_A$ is an $(n+1)\times s$ block, $\Gamma_A$ is an $r\times (d+1)$ block. The supermatrix is in $A_0^{M+1|N}$ if the blocks $Z_A$ and $T_A$ have entries in $A_0$ and the blocks $\Lambda_A$ and $\Gamma_A$ have entries in $A_1$.

 Let us denote the $2\times 2$ minor of this matrix, with rows $(k,l)$ and columns $(i,j)$, as  ${d_{kl}^{ij}}_A(M)$. The result is and element of $A$ with definite parity, 0 or 1. For each choice of $(i,j), (k,l)$ we can define a family of $A_0$-bilinear maps
$$\begin{CD}
A^{M+1|N}_0@>{d^{ij}_{kl}}_A>>A_{0,1}\\
M_A@>>>{d_{kl}^{ij}}_A(M_A)
\end{CD}$$
which are also functorial in $A$. If the target is $A_0$, we have that $A_0= (A\otimes \C^{1|0})_0$ and if the target  is $A_1$ we have that $A_1= (A\otimes \C^{0|1})_0$. Besides this remark, the argument does not change. The construction is functorial in $A$ and composing it with $f_A$ in (\ref{identitymaps}) we have families of maps
$$ \begin{CD}A^{n+1|r}_0\times A^{d+1|s}_0@>{d^{ij}_{kl}}_A\circ f_A>>  A_{0,1}\,, \end{CD}$$ and applying the even rules principle they define morphisms among the super vector spaces
$$\begin{CD}\C^{n+1|r}\otimes\C^{d+1|s}@>{d^{ij}_{kl}}\circ f>>\C^{1|0}, \C^{0|1}\,.\end{CD}$$

Let us now consider again two elements $a\otimes x+\theta\otimes\alpha$ and $b\otimes y+\xi\otimes\beta$. We write the image under $f_A$ in matrix form
\be\begin{pmatrix} ab\otimes (x\otimes y)&b\theta\otimes (\alpha\otimes y)\\
a\xi\otimes(x\otimes \beta)&-\theta\xi\otimes (\alpha\otimes \beta)\end{pmatrix}\,.\label{decomposablematrix}\ee
It is easy now to check that this matrix is in the kernel of ${d^{ij}_{kl}}_A$. On the other hand, since we are dealing with local algebras and ordinary $A_0$-modules, the condition
\be{d^{ij}_{kl}}_A(M_A)=0\,.\label{polynomials1}\ee
is, as for vector spaces, equivalent to say that the matrix $M$ is decomposable as (\ref{decomposablematrix}), so we have identified the image of the map $f_A$ as the solution to the polynomial equations (\ref{polynomials1}).
This condition can be translated into the super vector space morphism, were we get ride of the auxiliary variables in $A$. We dispose the generators of  $\C^{N+1|M}$ (as an affine superspace) in matrix form  similarly to (\ref{matrix}). We have that the image of the super Segre map is given by the homogeneous polynomials
\be{d^{ij}_{kl}}(M)=0,\qquad M=\begin{pmatrix}Z&\Lambda\\\Gamma&T\end{pmatrix}\,.\label{polynomials2}\ee
Since the polynomials are homogeneous, they are indeed constraints in projective space.

\begin{remark} Notice that the equations (\ref{polynomials2}) are determinants only for the upper left block. In all the other cases they pick up signs. The auxiliary variables in $A$ help us to keep track of these signs.

\hfill$\square$

\end{remark}

Once the relations (\ref{polynomials2}) are given, they provide us
with a supergraded ring $\C[Z,\Lambda,\Gamma,T]/d^{ij}_{km}$ defining
a projective superscheme in $\bP^{N|M}$. This superscheme is identified via
the morphism with $\bP^{n|r}\times \bP^{d|s}$,
as one can readily check on the standard covering.

\end{document}